\documentclass{amsart}
\usepackage{amsmath, amsthm, amssymb}
\usepackage{mathrsfs}
\usepackage{mathtools} 
\usepackage{hyperref} 
\usepackage{physics} 
\usepackage{enumitem} 
\usepackage{chngcntr} 
\usepackage{a4wide} 
\usepackage[tableposition=above]{caption} 
\usepackage{relsize} 
\usepackage{graphics}
\graphicspath{{./images/}}
\usepackage{subcaption}
\DeclareCaptionSubType[Alph]{table}
\usepackage[backend=biber, sorting=nyt, style=numeric]{biblatex} 
\usepackage{caption}
\DeclareNameAlias{sortname}{last-first}
\renewcommand{\arraystretch}{1.25}

\newcounter{theorem}
\newtheorem{thm}[theorem]{Theorem}
\newtheorem{lemma}[theorem]{Lemma}
\newtheorem{prop}[theorem]{Proposition}

\theoremstyle{definition}
\newtheorem{defn}[theorem]{Definition}

\newtheorem{rmk/}[theorem]{Remark}

\newtheorem{notation}[theorem]{Notation}

\newenvironment{rmk}
	{%
  	\pushQED{\qed}\begin{rmk/}} {\popQED\end{rmk/}}

\counterwithout{table}{section}
\counterwithin{equation}{section}
\counterwithin{theorem}{section}
\newcommand{\scr}[1]{\mathcal{#1}}
\newcommand{\goth}[1]{\mathfrak{#1}}
\newcommand{\bb}[1]{\mathbb{#1}}

\newcommand{\vp}{\varphi}
\newcommand{\ve}{\varepsilon}

\newcommand{\g}{\goth{g}}
\newcommand{\ga}{\gamma}

\newcommand{\D}{\Delta}
\newcommand{\limtau}{\lim_{\tau \to -\infty}}

\newcommand{\vi}{$\mathrm{VI}_{-1/9}$}
\newcommand{\ot}{$\mathrm{OT}$}
\newcommand{\otp}{\mathrm{OT}_{+}}
\newcommand{\otpo}{\mathrm{OT}_{+}^{\Omega}}
\newcommand{\bpvi}{\mathrm{B}_+}
\newcommand{\bpvio}{\mathrm{B}_+^{\Omega}}
\newcommand{\svi}{\mathrm{S}}
\newcommand{\spvi}{\mathrm{S}_+}
\newcommand{\spvio}{\mathrm{S}_+^{\Omega}}
\newcommand{\hpvi}{\mathrm{HP}_+}
\newcommand{\hpvio}{\mathrm{HP}_+^{\Omega}}
\newcommand{\advi}{\mathrm{AD}_+}
\newcommand{\advio}{\mathrm{AD}_+^{\Omega}}

\newcommand{\isvio}{\mathrm{IS}_+^{\Omega}}

\renewcommand{\S}{\Sigma}

\newcommand{\s}{\sigma}
\newcommand{\Lp}[1]{L^{#1}((-\infty,0])}

\newcommand{\rK}{\mathring{K}}

\newcommand{\ha}{\hat{a}}
\newcommand{\hn}{\hat{n}}
\newcommand{\hg}{\hat{\gamma}}
\newcommand{\he}{\hat{e}}
\newcommand{\ho}{\hat{\omega}}

\DeclareMathOperator{\Id}{Id}
\DeclareMathOperator{\roRic}{Ric}
\DeclareMathOperator{\ad}{ad}
\DeclareMathOperator{\cl}{cl}

\DeclareMathOperator{\roScal}{Scal}

\title{Quiescence for the exceptional Bianchi cosmologies}
\author{Hans Oude Groeniger}
\address{Department of Mathematics, KTH, 100 44 Stockholm, Sweden}
\email{joog@kth.se}
\date{}

\begin{document}
%
%
\begin{abstract}
The exceptional Bianchi type~VI$_{-1/9}$ cosmologies,
with non-stiff fluid matter or in vacuum, are conjectured to generically undergo
an infinite series of chaotic oscillations toward the initial singularity,
as has been proven to occur for Bianchi type VIII and IX cosmologies.
Cosmologies of the lower Bianchi types, i.e. except those of type VIII or IX, 
admit a two-dimensional Abelian subgroup of the isometry group, the~$G_2$.
In orthogonal perfect fluid cosmologies of all lower Bianchi types
except for type~VI$_{-1/9}$ the~$G_2$ acts orthogonally-transitively,
which is closely related to an eventual cessation of the oscillations
    and thus to a quiescent singularity.
But due to a degeneracy in the momentum constraints,
such cosmologies of type VI$_{-1/9}$ do not necessarily have this property.
As a consequence, the dynamics of type~VI$_{-1/9}$ orthogonal perfect fluid cosmologies
have the same degrees of freedom as those of the higher types~VIII and~IX 
and their dynamics are expected to be markedly different compared to those of 
the other lower Bianchi types.

In this article we take a different approach to quiescence,
namely the presence of an orthogonal stiff fluid.
On the one hand, this completes the analysis of the initial singularity 
for all Bianchi orthogonal stiff fluid cosmologies.
On the other hand, this allows us to get a grasp of the underlying dynamics
of type VI$_{-1/9}$ perfect fluid cosmologies, 
in particular the effect of orthogonal transitivity 
as well as possible (asymptotic) polarization conditions.
In particular, we show that a generic type~VI$_{-1/9}$ cosmology 
with an orthogonal stiff fluid has similar asymptotics
as a generic Bianchi type VIII or IX cosmology with an orthogonal stiff fluid.
The only exceptions to this genericity are solutions satisfying
    (asymptotic) polarization conditions,
and solutions for which the~$G_2$ acts orthogonally-transitively.
Only in those cases may the limits of the eigenvalues
    of the expansion-normalized Weingarten map be negative.
We also obtain a concise way to represent the dynamics which works more
generally for the exceptional type~VI$_{-1/9}$ orthogonal perfect fluid cosmologies,
and obtain a monotonic function for the case of a orthogonal perfect fluid 
that is more stiff than a radiation fluid.
\end{abstract}
\maketitle
%
%
\section{Introduction} \label{sec:Introduction}
\noindent In this article we consider spacetimes solving Einstein's equations
with fluid matter or in vacuum that have initial big bang singularities.
We call such spacetimes \emph{cosmologies}.
Moreover, the results we obtain are about Bianchi cosmologies with orthogonal stiff fluid matter.
In particular we consider the \emph{exceptional} Bianchi type~\vi{},
which is of Bianchi class B, unlike the well-known Bianchi types VIII and IX.
Bianchi cosmologies of type VIII or IX play a central role in the framework of 
Belinski, Khalatnikov and Lifschitz
and their collaborators (BKL), see e.g. \cite{BKLTimeSingularity1982} and references therein,
as they are the simplest models that we know of that exhibit an oscillatory big bang singularity.
By this we mean, loosely speaking,
that toward the big bang singularity the spacetime undergoes
an infinite sequence of contractions and expansions in different directions of spacetime.
This behaviour is captured by the expansion-normalized Weingarten map $\scr{K}$,
which is the endomorophism metrically equivalent 
    to the second fundamental form normalized by dividing with the mean curvature.
More specifically, given some preferred foliation, in the BKL framework
the eigenvalues of $\scr{K}$
generically do not converge and their dynamics closely follow the so-called Kasner map.
For cosmologies of type VIII and IX,
rigorous results have been established in this direction,
see in particular \cite{RingstromBianchiIX,HeinzleUgglaNew} 
and more recently \cite{BeguinDutilleul}.
The exceptional Bianchi type~\vi{} cosmologies with a non-stiff orthogonal perfect fluid or in vacuum
are conjectured to have an oscillatory singularity
    similar to the Bianchi type VIII or IX cosmologies.

The complement to these oscillatory big bang singularities
are \emph{quiescent} big bang singularities,
for which the eigenvalues of the expansion-normalized Weingarten map $\scr{K}$
of a preferred foliation converge toward the singularity. 
A pioneering result in this direction is the article \cite{AnderssonRendall} by Andersson and Rendall.
In this article analytic solutions to the Einstein-scalar field
and Einstein-stiff fluid equations are constructed, using so-called Fuchsian methods, 
by essentially prescribing the behaviour of the solution at the singularity.
This idea of understanding the solution by its properties at the singularity
has been formalized by Ringström in \cite{RingstromInitialDataSingularity},
and specialized to cosmologies with Bianchi class A or Gowdy symmetry 
    in \cite{RingstromInitialDataSingSym},
leading to the notion of \emph{initial data on the singularity}.
Other notable examples of this type that are relevant to the discussion here, are
found e.g. in the article by Isenberg and Moncrief 
\cite{IsenbergMoncriefHalfPolarized} and 
more recently in \cite{FournodavlosLukAsymptoticallyKasner} by Fournodavlos and Luk.

A more recent class of results is that
of the stable formation of quiescent big bang singularities outside of symmetry.
Here stable refers not to the nonlinear stability of a specific solution,
but rather to the nature of the formation of the singularity.
Firstly there are the articles \cite{RodnianskiSpeckNearFLRW,
SpeckIsotropicBianchiIX, FajmanUrbanBianchiV, BeyerOliynykPastStabilityScalarField,
BeyerOliynykPastStabilityFluids},
in which stable big bang formation is shown for developments of initial data that 
are close to being both isotropic and spatially (locally) homogeneous,
in the presence of scalar field or stiff fluid matter, 
or in vacuum but then necessarily in high dimension.
Note also \cite{RodnianskiSpeckModeratelyAnisotropic}, 
for the case with mild anisotropy in vacuum, but in high dimension.
The results obtained in \cite{FournodavlosRodnianskiSpeck} are of a similar nature,
but concern data close to that induced by solutions
that are potentially highly anisotropic but spatially homogeneous and spatially flat
(namely the analogues of the Kasner solutions in the Einstein-scalar field setting).
On the other hand, there is the article in \cite{QuiescentBigBang} 
where instead of closeness to a particular solution, 
a criterion is given concerning the size of the mean curvature that
depends in particular on various bounds on expansion-normalized quantities
constructed from the initial data, including $\scr{K}$ and its eigenvalues.
Then stable formation of quiescent big bang singularities 
    follows for developments of initial data satisfying the criterion.
This result allows one to make conclusions, which are analogous to the other results mentioned above,
but for developments of initial data close to that induced by
    a large class of spatially locally homogeneous solutions,
in particular also those that are not close to isotropy or nearly spatially flat.
Moreover, the criterion places no direct restrictions on the inhomogeneity of the data.

A shortfall of many of the results on stable big bang formation
is that there are no significant dynamics.
What essentially happens is that a regime is identified
in which the dynamics of Einstein's equations are suppressed,
usually made possible by either stiff fluid or scalar field matter 
    or a high dimension of spacetime.
Therefore, initial data close to a model solution has properties
    similar to the model solution,
and this allows for the stable formation of big bang singularities.
On the other hand, there are many results concerning cosmologies with Bianchi or Gowdy symmetry,
which have quiescent big bang singularities
toward one time direction and significant dynamics toward the other time direction,
both in vacuum or with various types of matter.
Focusing our attention of Bianchi cosmologies,
we note in particular \cite{WainwrightHsuClassA, HewittWainwrightClassB, HewittTraceClass, RendallMixmaster, RingstromBianchiIX, LeBlancKerrWainwrightMagneticVI0, HewittBridsonWainwrightTiltedII, HervikTiltedVI0, RadermacherVacuum, RadermacherStiff, OudeGroenigerVI0}, 
among many others.
The results in these articles are interesting for various reasons,
but in the current context 
for mainly two considerations.
First, Bianchi cosmologies with quiescent singularities may be very close,
at any stage of their development, to those that have oscillatory singularities.\footnote{
Closeness here is to be understood in the sense of the phase space of the expansion-normalized variables
such as given e.g. in \cite{WainwrightHsuClassA}.}
This is because quiescence may occur not only as result of dynamics being suppressed
    by the presence of certain types of matter,
which we note is necessary for the results on stable formation of big bang singularities in four dimensions.
It may and does also occur as a result of geometrical features related to symmetries.
Second, the intermediate evolution and future asymptotics of Bianchi cosmologies
can be analyzed to a larger level of detail compared to inhomogeneous cosmologies 
due to the field equations being described effectively by a set of ordinary differential equations.

Thus, quiescence of the initial singularity can be caused by the dynamics
being suppressed due to the presence of certain types of matter,
or due to a high dimension of spacetime.
But, as mentioned above,
it may also occur due to the presence of certain geometrical features due to symmetry.
The heuristics of this were already brought forward by BKL,
but may formally be understood 
using the framework developed in \cite{RingstromGeometrySilentBigBang}.
Although the results in \cite{RingstromGeometrySilentBigBang}
rely on certain assumptions concerning the expansion-normalized Weingarten map $\scr{K}$
of a preferred foliation
and concerning its expansion-normalized normal derivative,
the crux of the matter may be understood without going too much into the details.

Here we are interested in the four-dimensional case, 
so consider an orthonormal $\{e_i\}_{i=1}^3$ frame diagonalizing~$\scr{K}$.
Also, order them so that the corresponding eigenvalues $p_i$ satisfy $p_1 \leq p_2 \leq p_3$.
Then there are two sufficient conditions that lead to quiescence.
Either all the $p_i$ are positive along the foliation all the way up to the initial singularity,
or if $p_1$ is negative then the structure coefficient $\gamma^1_{23}$ associated with this frame
must vanish identically along the foliation.
The former condition is the one that occurs only for specific types of matter,
most notably stiff fluids and scalar fields, or a high dimension
and in four-dimensional spacetimes is incompatible with vacuum,
due to an asymptotic version of the Hamiltonian constraint.
The latter condition is non-generic and caused by geometrical features due to symmetries.
There are moreover examples, among which in this article, where the structure coefficient
does not vanish identically but only asymptotically, and the presence of the geometrical feature
may be understood to be only at the initial singularity, the so-called half-polarized solutions.

There are two notable examples of such geometrical features.
The first examples has to do with the presence of a two-dimensional Abelian subgroup of the isometry group,
the Abelian $G_2$, which is present in all the Bianchi cosmologies of lower type, 
i.e. except for those of Bianchi type~VIII or~IX.
The action of the $G_2$ may have the property of orthogonal transitivity,
which means that the planes in spacetime orthogonal to the $G_2$ are the tangent planes of some surface.
If this is the case, then the second condition, concerning the vanishing of a structure coefficient,
may be shown to be satisfied.
The second example is when one of the Killing fields generated by the isometry group 
is hypersurface-orthogonal. 
This also leads to the vanishing of a structure coefficient
for the frame diagonalizing the Weingarten map.
Also the presence of a hypersurface-orthogonal Killing vector field 
is something that is typically non-generic within the symmetry class,
which in the results below is evidenced by the non-genericity of the orbits satisfying this extra requirement.
We expand on both of these notions in Appendix~\ref{sec:abelianG2}.

In the family of spacetimes under consideration in this article,
namely Bianchi type~\vi{} cosmologies with orthogonal stiff fluid matter,
both of these causes of quiescence are present,
and both examples of geometrical features occur in this family as well.
The main result is that the behaviour associated with quiescence of the former type is generic,
i.e. generically the quiescence is related to the stiff fluid 
and the positivity of all eigenvalues of the expansion-normalized Weingarten map. 
For a null set of orbits the behaviour associated with quiescence of the second type occurs,
i.e. there are examples where the quiescence is related to geometrical features. 
We find in particular that not all limits of the eigenvalues of the
expansion-normalized Weingarten map are necessarily positive,
but the set of orbits for which they are strictly positive is generic.
This may be contrasted with the results concerning stable big bang formation mentioned above.
In those cases, the eigenvalues of the expansion-normalized Weingarten map are positive initially
and remain positive throughout the past development.

The methods used in this article are in the tradition of the dynamical systems approach.
The dynamical systems approach, in particular the expansion-normalized variant
of Wainwright and Hsu \cite{WainwrightHsuClassA} and subsequently developed by many others,
based on the tetrad or orthonormal frame approach of Ellis and MacCallum \cite{EllisMacCallumHomCos},
has been fruitful in the analysis of spatially homogeneous cosmologies.
Most studies to date have focused on Bianchi class A cosmologies,
in part due to its favourable properties in choosing an orthonormal frame;
in Bianchi class A orthogonal perfect fluid cosmologies, one may simultaneously diagonalize
the Weingarten map and a symmetric matrix associated to the structure coefficients
(as well as the spatial Ricci and Cotton-York tensors).
For the Bianchi class B orthogonal perfect fluid cosmologies the 
fact that the $G_2$ does not act necessarily orthogonally-transitively forms an obstruction to this approach.
It adds another degree of freedom and any choice of frame results in a trade-off
on how to describe this additional degree of freedom.
This is reflected in additional variables and constraints
in the description in expansion-normalized variables.
For this reason the behaviour of type~\vi{} orthogonal perfect fluid
cosmologies is considered as the generic one
within the class B orthogonal perfect fluid cosmologies,
similar to the behaviour of Bianchi types VIII and IX cosmologies within the class A cosmologies.
In fact, orthogonal perfect fluid cosmologies of types VIII and IX
have the same number of degrees of freedom as those of type~\vi{}.
One could say that the analysis of Bianchi class A cosmologies with orthogonal perfect fluids
obscures complications related to the frame.
Resolving these complications are important to obtain a better grasp of quiescent singularities.
It should be noted, however, that if the matter has some built in anisotropy,
then these types of complications also occur for Bianchi class A.
This happens for example for tilted fluids, cf. e.g. \cite{HewittBridsonWainwrightTiltedII}
and \cite{HervikTiltedVI0}, or for magnetic Bianchi cosmologies,
cf. e.g. \cite{LeBlancKerrWainwrightMagneticVI0} and \cite{WeaverMagnetic}.
The approach with respect to the frame chosen in \cite{HewittHorwoodWainrightExceptional} 
    is adapted to $G_2$ cosmologies;
in particular one of the frame vectors is chosen orthogonal to the Abelian $G_2$, 
and the other two tangent to it. 
There is a remaining freedom of rotating the two frame vectors tangent to the $G_2$,
which there is used to eliminate one degree of freedom in the shear.
The additional degree of freedom of the exceptional  
type~\vi{} cosmologies is then manifested in one additional off-diagonal shear component,
as well as one additional off-diagonal component of the structure coefficients.
Our approach is similar, but diagonalizes the structure coefficients,
leaving all the additional degrees of freedom to be manifested in the shear.
Moreover, we introduce a double set of polar coordinates which allows 
to write one of the off-diagonal shear terms as a graph of the other functions,
thereby getting rid of one of the constraints.


\subsection{Overview of results} \label{sec:OverviewOfResults}
The main result of this article, Theorem~\ref{thm:GenericConvergenceToTriangle},
may be read as follows:
\emph{For a generic Bianchi type~\vi{} cosmology with orthogonal stiff fluid matter,
the initial singularity is quiescent, anisotropic and contracting in all directions.
In particular, the limits of the eigenvalues of the expansion-normalized Weingarten map
exist, are distinct and are all strictly positive.}

The word \emph{generic} here means that the statement holds
for a full measure subset of orbits in the expansion-normalized phase space,
so with respect to the expansion-normalized phase space it is 
Lebesgue generic.
We moreover show that the generic set 
    consists of a countable intersection of open and dense subsets,
    so the set is also Baire generic.
There are two sets of orbits that fall outside of this set of generic orbits.

There are the orthogonally-transitive orbits, in which the $G_2$
does act orthogonally transitive.
In this case the limits of the eigenvalues of the expansion-normalized Weingarten map
need not all be strictly positive; the eigenvalue corresponding to the direction
orthogonal to the $G_2$ may be negative.
This is a similar situation as happens for other class B cosmologies with 
an orthogonal stiff fluid, 
which have been studied in different expansion-normalized coordinates in \cite{RadermacherStiff}.

But there is also the subset of \emph{half-polarized} orbits.
The orbits in this set, which are non-generic in the sense above,
have asymptotic behaviour similar to the polarized orbits
which have a hypersurface-orthogonal Killing vector field.
Therefore we call them half-polarized as they arise in an analogous way to that of the half-polarized, 
$U(1)$-symmetric solutions to the Einstein-vacuum equations studied
in \cite{IsenbergMoncriefHalfPolarized} using Fuchsian methods;
beyond polarized asymptotic data, 
where polarization means that the Killing vector field generating the symmetry
is hypersurface-orthogonal, 
the authors there are able to specify asymptotic data with another degree of freedom
that leads to solutions with similar behaviour
as the solutions arising from polarized asymptotic data,
the so-called half-polarized solutions.
Also in this case the smallest of the limits of the eigenvalues of the expansion-normalized
Weingarten map may be negative.
Similar solutions also arise in other class B cosmologies with orthogonal perfect fluids, 
cf. \cite{HewittTraceClass}, \cite{HewittWainwrightClassB},
\cite{RadermacherVacuum} and \cite{RadermacherStiff},
as for those the $G_2$ acts orthogonally-transitively.

Finally, we also obtain a result regarding the expanding direction.
A monotone quantity is derived for Bianchi type~\vi{} cosmologies
with orthogonal perfect fluid matter with linear equation of state
of the form $\rho = (\gamma - 1) p$ for $\gamma \in [4/3, 2]$, 
which allows for the conclusion that 
the expansion-normalized matter density vanishes asymptotically for this range of $\gamma$.

\subsection{Comparison to other Bianchi stiff fluid cosmologies}
An interesting contrast of the results described here
is with the class A cosmologies with an orthogonal stiff fluid,
due to the complications with the frame that do not arise in class A.
These results concerning class A cosmologies 
have been described in Section~7 of \cite{RingstromBianchiIX}.
For such cosmologies of the higher Bianchi types VIII and IX,
the limits of the eigenvalues of the expansion-normalized
Weingarten map are always strictly positive, and there is no analogue
of the (half-)polarized orbits.
For such cosmologies of Bianchi type VI$_0$ and VII$_0$,
the behaviour is similar to the orthogonally-transitive orbits of type~\vi{},
as a negative eigenvalue may occur but the corresponding eigenspace 
is then orthogonal to the Abelian $G_2$.
In Bianchi type VI$_0$ with an orthogonal perfect fluid, there is also a set similar 
to the polarized orbits, cf. \cite{DynSysCosmology, OudeGroenigerVI0},
where it is called the shear-invariant set.
For Bianchi type II, there is not only one Abelian $G_2$
but we can find two $G_2$'s that act orthogonally transitively
in the case of an orthogonal perfect fluid.
Also, if one of the limits of the eigenvalues of the expansion-normalized
Weingarten map is negative, then its corresponding eigenspace is orthogonal to both of these.
Lastly for Bianchi type I orthogonal perfect fluid cosmologies, any two shear eigenvectors commute
and form an orthogonally transitive $G_2$, and the limit of the eigenvalue being negative
can occur for any eigenspace of the shear.

The non-exceptional class B cosmologies with orthogonal stiff fluid matter
    have been studied in detail in \cite{RadermacherStiff}.
When it comes to the polarized and half-polarized solutions, similar phenomena
as described above for type~\vi{} occur for Bianchi types IV and VI$_\eta$.
However, the behaviour for generic orbits of type IV,
    VI$_\eta$ ($\eta \neq -1/9)$ and VII$_\eta$,
for all of which the $G_2$ acts orthogonally-transitively,
is similar to the behaviour of the orthogonally-transitive orbits of type~\vi{}.
This is because that not all limits of the eigenvalues of
the expansion-normalized Weingarten map are 
necessarily positive, and if indeed one is negative
then its eigenspace corresponds to the direction orthogonal to the $G_2$.

\subsection{Outline of the paper}
In Section~\ref{sec:Equations} we introduce the evolution equations that we use,
which are in particular specialized for the orthogonal stiff fluid case.
The evolution equations are derived in Appendix~\ref{sec:AppendixEquations} 
from the usual equations as found e.g. in \cite{DynSysCosmology} 
using the well-known orthonormal frame formalism.
We moreover introduce the phase space and discuss the non-vacuum fixed points.

In Section~\ref{sec:matterdensity} we do a preliminary analysis,
mainly using the monotonicity of the expansion-normalized matter density,
and derive convergence of the backward orbits.
In Section~\ref{sec:asympdiag} we consider the main class of orbits,
namely the \emph{asymptotically diagonalized} orbits,
which is shown to be generic in Section~\ref{sec:halfpolarized}.
These solutions converge to specific sections of the Jacobs disc toward the past.
Next, in Section~\ref{sec:halfpolarized}, we consider the so-called half-polarized solutions,
which toward the singularity behave similarly to solutions having the additional
symmetry that one of the Killing fields is hypersurface-orthogonal.
We show that this class of solutions is non-generic within the set of all Bianchi type~\vi{}
orthogonal stiff fluid solutions.
The statement and proof of the main result are contained in Section~\ref{sec:MainTheorem}.

Section~\ref{sec:future} contains some observations regarding
the expanding direction also for non-stiff Bianchi type~\vi{} cosmologies.
In particular, we find a monotone function for the range $\gamma \in [4/3,2]$.
We conclude with several appendices.
Firstly, there is an appendix where we recall some important properties of the Bianchi classification,
after which we make use of this to discuss the relevant geometric features of the Abelian $G_2$
and the hypersurface-orthogonal Killing vector field.
Next, Appendix~\ref{sec:AppendixEquations} contains the derivations for the equations
which are not specialized to the stiff fluid,
and lastly Appendix~\ref{sec:dynsys} contains some analytical tools 
that prove useful for the analysis.

We assume the reader is familiar with the orthonormal-frame approach 
of \cite{EllisMacCallumHomCos}, as well as with various aspects of the dynamical systems approach
to cosmology which can be found in \cite{DynSysCosmology}.
We also use geometrized units $c=1$ and $8\pi G = 1$.
Moreover, Latin indices run from 1 to 3 and Greek indices run from 0 to 3.

%
%

\section{Expansion-normalized evolution equations} \label{sec:Equations}
\noindent Specializing the system of equations \eqref{eq:PolarEvEAppendix}
as well as the accompanying constraints found in Appendix~\ref{sec:AppendixEquations}
to the case of a stiff fluid, so in particular $\gamma = 2$ and $q^*=2$,
we obtain the following system of equations:
\begin{subequations} \label{eq:PolarEvE}
\begin{align} 
M' &= \big[q +2 \S_{+} + 2 \sqrt{3} \cos(\theta) \S_- \big] M, \label{eq:EvEM}\\
P' &= \big[-(2-q) - \big(3 + \sin^2(\theta) \big) \S_{+}
    + \sqrt{3} \cos(\theta) \S_-  \big] P, \label{eq:EvEP} \\
\theta' &= -2 \sqrt{3} \S_- \sin(\theta), \label{eq:EvETh}\\
\S_{+}' &= -(2-q) \S_{+} - 2 M^2 + 3 P^2, \label{eq:EvESplus} \\
\S_{-}' &= -(2-q) \S_- - \sqrt{3} \cos(\theta) \big( 2 M^2 + P^2 
    - \tfrac{2}{3} \sin^2(\theta) \S_{+}^2 \big) \label{eq:EvESminus},
\intertext{where $q$ is now shorthand for}
q &= 2-2 \big(1+ \tfrac{1}{3}\sin^2(\theta)\big) M^2.
    \label{eq:DefinitionOfq}
\intertext{We also have the constraints defining $\Omega$ and $\S_{\times}$, namely}
\Omega &= 1 - \big(1+ \tfrac{1}{3}\sin^2(\theta)\big) \big( M^2 + \S_+^2 \big) 
    - P^2 - \S_-^2, \label{eq:HCNew} \\
\S_{\times} &= \tfrac{1}{\sqrt{3}} \S_+ \sin (\theta), \label{eq:MCSx}
\intertext{and the corresponding auxiliary equations}
\S_{\times}' &= [-(2-q) - 2 \sqrt{3} \S_- \cos(\theta)] \S_{\times} 
    - \tfrac{1}{\sqrt{3}} \sin(\theta) \big( 2 M^2 - 3 P^2 \big), \\
\Omega' &= -2(2-q) \Omega. \label{eq:EvEOm}
\end{align}
\end{subequations}

\subsection{The phase space and fixed points on the boundary}
We are now in a position to give a description of the phase space
of the system we wish to consider.
In particular, we define the set $\bpvi$ as follows:
\begin{equation}
\bpvi :=
    \left\{(M, P, \theta, \Sigma_+, \Sigma_-) \in \bb{R}^5 ~|~
		 M > 0, P \geq 0,  \theta \in (0, \pi) \text{~and~} \Omega \geq 0
\right\},
\end{equation}
and refer to it as the \emph{phase space}.
The closure $\cl \bpvi$, which in particular consists of the boundaries
for which $M = 0$ and for which $\sin(\theta) = 0$, 
which respectively describe the type I and type II orbits,
we refer to as the \emph{extended phase space}.

Moreover, let $\bpvio := \bpvi \cap \{ \Omega >0 \}$
denote the non-vacuum orbits
and let $\bpvi^{0} := \bpvi \cap \{\Omega = 0\}$
denote the vacuum orbits, both of which are invariant subsets of $\bpvi$.

Moreover, we denote
\begin{equation}
\otp := \bpvi \cap \{P =0\},
\end{equation}
for the orthogonally-transitive orbits in the stratum $\bpvi$, as well as 
\begin{equation}
\spvi := \bpvi \cap \{\cos(\theta) = 0, \S_- = 0\},
\end{equation}
for the polarized orbits in the stratum $\bpvi$.
We also define $\otp^{0}, \otpo$ and $\spvi^0, \spvio$ in an analogous way
to distinguish the vacuum and non-vacuum orbits.

\begin{rmk} \label{rmk:SmoothEquivalence}
Due to the homogeneous nature of Equations~\eqref{eq:EvEM} and~\eqref{eq:EvEP},
the subsets $\{M=0\}$ as well as $\{P=0\}$ of the extended phase space $\cl \bpvi $
are invariant sets under the evolution.
In particular, outside of these sets, the coordinate transformations used
to bring the original system~\eqref{eq:CanonEvE} into the form~\eqref{eq:PolarEvE} above
is a smooth map for orbits with $M > 0$ and $P >0$,
and we thus obtain a smooth equivalence between the dynamics.
For orbits for which $M > 0$ but $P  = 0$ we obtain such a statement as well,
as we may regard the variable $\phi = \theta/2$
as being eliminated through the constraints.
The only issue to relate orbits of \eqref{eq:CanonEvE} with those of \eqref{eq:PolarEvE}
    is for orbits satisfying $M = 0$, which are type I orbits.
The type I orbits be related using the eigenspaces and eigenvalues of the expansion-normalized shear.
But as our main interest is in making conclusions about type~\vi{} orbits,
i.e. those lying inside the phase space $\bpvi$,
it suffices for the conclusions we wish to make 
that the orbits are smoothly equivalent for $M >0$.
\end{rmk}

\begin{notation} Given $x \in \cl{\bpvi}$
we let $(M(x),P(x), \theta(x), \S_+(x), \S_-(x)) = (x_1,...,x_5)$.
Moreover, given initial data $x \in \bpvi$ we write e.g.
\begin{equation*}
M(\tau) := \big( M \circ \vp^\tau \big) (x) := \big(\vp^\tau(x)\big)_1
\end{equation*}
and similarly for the other coordinates and functions of the coordinates.
Here and throughout this article $\vp^\tau$ denotes the flow of the system 
defined by \eqref{eq:PolarEvE}.
\end{notation}

\begin{rmk} \label{rmk:alphalimit}
Due to the requirement that $\Omega \geq 0$,
i.e. by the Hamiltonian constraint \eqref{eq:HCNew},
we observe that $M, P, \sin(\theta), \S_+^2, \S_-^2$ as well as $\Omega$ lie in the interval $[0,1]$.
Hence, the closure of the phase space is compact.
Moreover, as the derivatives in $\tau$ are also bounded
the flow $\vp^\tau$ is complete.

Also, recall that $\alpha$- and $\omega$-limit sets are closed and invariant. 
Since the orbits of our dynamical system are contained in a compact invariant set,
any limit set of a point in $\cl{\bpvi}$ is non-empty and connected,
see e.g. Proposition~1.1.14 of \cite{WigginsNonlinearDynSys}.
\end{rmk}

\subsubsection{Non-vacuum fixed points}
The flow of the system \eqref{eq:PolarEvE} may contain various fixed points
in either the phase space or its boundary.
Observe that any equilibrium point must satisfy either $q = 2$ or $\Omega =0$,
as a consequence of the evolution equation for $\Omega$.

We may compute that a non-vacuum fixed point in $\bpvio$ with coordinates
$(m, p, \psi, s_+, s_-)$, and $\Omega = \mu > 0$, must satisfy $q = 2$ and therefore $m = 0$. 
Then from the evolution equation for $\S_+$ we deduce that $p = 0$.
Next, we see that either $s_- = 0$ or $\sin(\psi) = 0$ as else $\theta' \neq 0$.
In the former case, we obtain that moreover $\cos(\psi)\sin^2(\psi)s_+^2 = 0$
else $\S_-' \neq 0$. Thus one of the following three must hold: $\sin(\psi) = 0$ or 
$s_- = 0$ and $\cos(\psi) = 0$
or $s_- = 0$ and $s_+= 0$.

In particular, we have four sets of non-vacuum fixed points. 
Firstly there are the two sets $\scr{D}^\pm$ of fixed points whose projections
in the $\S_+\S_-$-plane are the same and known as the Jacobs disc~$\scr{D}$,
i.e. $\scr{D}:=\{(\S_+,\S_-): \S_+^2 + \S_-^2 < 1\}$.
These fixed points correspond to the Jacobs stiff fluid solutions,
and its boundary forms the Kasner circle $K := \{(\S_+,\S_-): \S_+^2 + \S_-^2 = 1\}$
    corresponding to the vacuum Kasner solutions.
Depending on whether $\cos(\theta) = 1$, or $\cos(\theta) = -1$
we denote the sets by $\scr{D}^+$ and $\scr{D}^-$ respectively.
Next, there is the set of fixed points $\scr{F}$ corresponding to the FLRW solutions
(this is now a set due to a degeneracy of adapted polar coordinates at that point,
cf. Remark~\ref{rmk:SmoothEquivalence}). 
Lastly there is the line of fixed points $\scr{P}$
containing the limits of the polarized orbits for which $\sin(\theta) = 1$.
In particular, these limits lie on the boundary of the invariant set $\spvi$.
The four sets are given respectively by
\begin{subequations}
\begin{align}
\scr{D}^+ &:= \left\{(0,0, 0, \S_+, \S_-) ~|~ \S_+^2 + \S_-^2 <1 \right\}, \\
\scr{D}^- &:= \left\{(0,0, \pi, \S_+, \S_-) ~|~ \S_+^2 + \S_-^2 <1 \right\}, \\
\scr{F} &:= \left\{ (0,0, \theta, 0, 0) ~|~ \theta \in [0,\pi] \right\}, \\
\scr{P} &:= \left\{ (0,0, \pi/2, \S_+, 0) ~|~ \S_+^2 < 3/4 \right\}.
\end{align}
\end{subequations}
Note that $\scr{D}^+, \scr{D}^-$ and $\scr{P}$ are connected through $\scr{F}$.
In particular, $\scr{D}^\pm \cap \scr{F} \neq \emptyset$
as well as $\scr{P} \cap \scr{F} \neq \emptyset$.

\begin{rmk} \label{rmk:ReconstructionOfSigma}
From the variables $M, P, \theta, \S_+, \S_-$, 
we can reconstruct the symmetric matrix $(\S_{ij})_{ij}$ which the expansion-normalized shear
written with respect to a canonical basis, cf. Appendix~\ref{sec:AppendixEquations},
in particular \eqref{eq:SigmaPlusSigmaMinus}.
Indeed, we have
\begin{align*}
\S_{11} &= - \tfrac{2}{3} \S_{+}, \quad 
\S_{22} = \tfrac{1}{3} \left (\S_+ + \sqrt{3} \S_- \right), \quad 
\S_{33} = \tfrac{1}{3} \left (\S_+ - \sqrt{3} \S_- \right), \\
\S_{23} &= \tfrac{1}{3} \S_+ \sin(\theta), \quad
\S_{31} = \tfrac{1}{\sqrt{3}} P \cos(\theta/2), \quad
\S_{12} = \tfrac{1}{\sqrt{3}} P \sin(\theta/2).  
\end{align*}
Note that for points in $\scr{D}^\pm$, only the diagonal elements are non-zero.
\end{rmk}
\begin{figure}[t]
     \centering
     \begin{subfigure}[b]{0.45\textwidth}
         \centering
         \includegraphics[width=\textwidth]{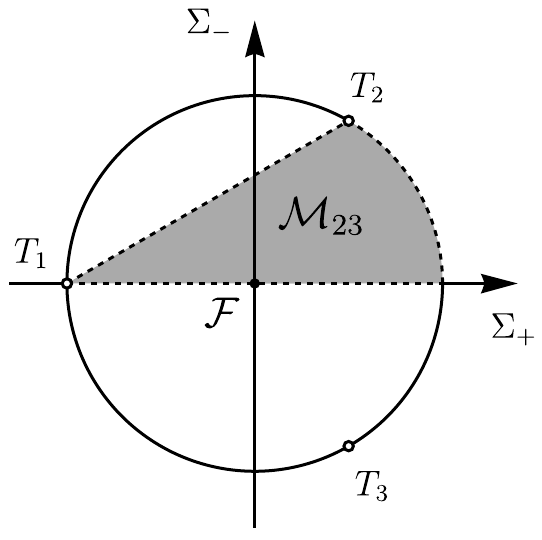}
     \end{subfigure}
     \begin{subfigure}[b]{0.45\textwidth}
         \centering
         \includegraphics[width=\textwidth]{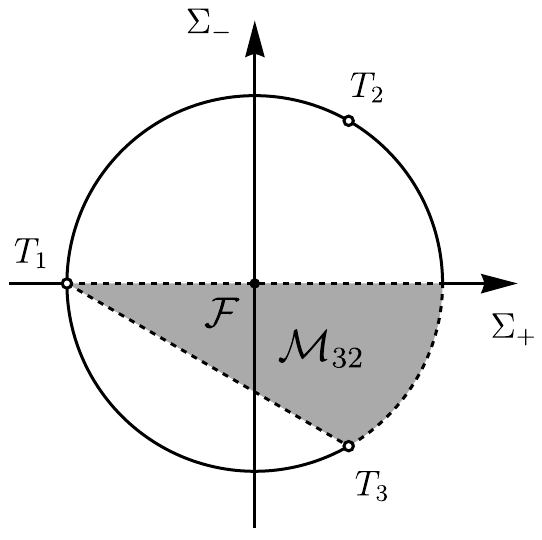}
     \end{subfigure}
        \caption{
        The projections of the sets $\scr{M}_{23}^\pm$ (left)
        and $\scr{M}_{32}^\pm$ (right)
        onto the Jacobs disc~$\scr{D}$ in the $\S_+ \S_-$-plane.
        Note that the boundaries of the regions are not part of the respective sets.
        The projection of the set $\scr{F}$, in the center of the circle, is also denoted,
        as well as the Taub points $T_i$, $i \in \{1,2,3\}$.}
        \label{fig:Mij}
\end{figure}
\begin{figure}[b]
     \centering
     \begin{subfigure}[b]{0.45\textwidth}
         \centering
         \includegraphics[width=\textwidth]{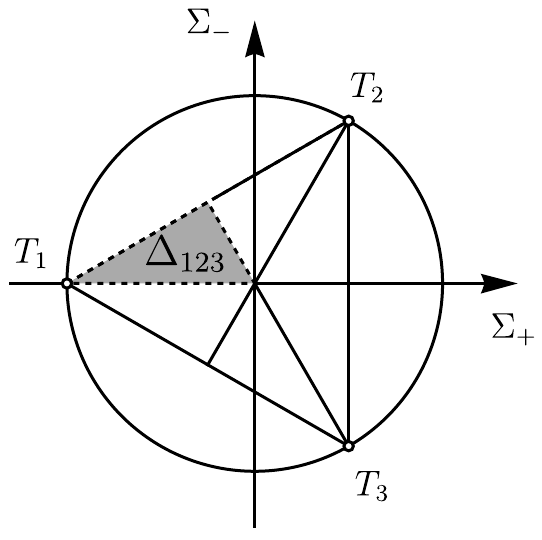}
     \end{subfigure}
     \begin{subfigure}[b]{0.45\textwidth}
         \centering
         \includegraphics[width=\textwidth]{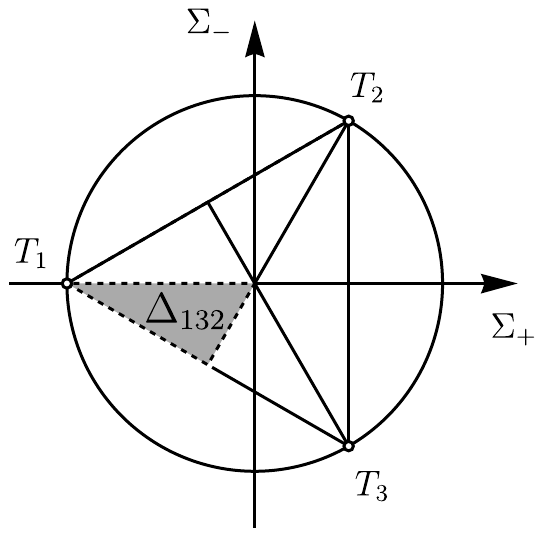}
     \end{subfigure}
        \caption{The projections of the sets
        $\D_{123}^\pm$ (left) and $\D_{132}^\pm$ (right)
        onto the Jacobs disc $\scr{D}$ in the $\S_+ \S_-$-plane.
        Note that the boundaries of the regions are not part of the respective sets.
        The triangle, i.e. the three line segments corresponding
        to one of the eigenvalues of the expansion-normalized Weingarten map vanishing,
        is also depicted, as well as the Taub points $T_i$, $i \in \{1,2,3\}$ lying at their intersections.}
        \label{fig:Dijk}
\end{figure}
It will prove to be convenient to further split $\scr{D}^\pm$ into several parts.
Firstly, given $i,j~\in~\{1,2,3\}$ such that $i \neq j$, 
we denote by $\scr{M}^\pm_{ij}$ the following subsets of $\scr{D}^\pm$,
namely
\begin{equation}
\scr{M}^\pm_{ij} := \scr{D}^\pm \cap \{ \S_{ii} > \S_{jj} > -\tfrac{1}{3}\},
\end{equation}
where $\S_{11}, \S_{22}, \S_{33}$ can be reconstructed from $\S_+$ and $\S_-$ 
through by Remark~\ref{rmk:ReconstructionOfSigma} above.
In particular, we have
\begin{align*}
\scr{M}^-_{23} &= \left\{(0,0,\pi,\S_+, \S_-) ~|~ 
    \S_+^2 + \S_-^2 < 1 \text{~and~}
    \S_+ - \sqrt{3} \S_- > -1 \text{~and~} 
    \S_- > 0 \right\}, \\
\scr{M}^+_{32} &= \left\{(0,0,0,\S_+, \S_-) ~|~ 
    \S_+^2 + \S_-^2 < 1 \text{~and~} 
    \S_+ + \sqrt{3} \S_- > -1 \text{~and~} 
    \S_- < 0 \right\},
\end{align*}
both of which play a role in the subsequent analysis.
In Figure~\ref{fig:Mij} the projections of $\scr{M}_{23}^\pm$ and $\scr{M}_{32}^\pm$
onto the Jacobs disc $\scr{D}$, i.e. in the $\S_+ \S_-$-plane, are shown.

On the other hand we denote by $\D^{\pm}_{ijk}$ certain sets whose projection into the $\S_+\S_-$-plane
lies inside what is commonly referred to as
    the (inside of the) \emph{triangle} inside the Jacobs disc,
which corresponds to the expansion-normalized Weingarten map being positive definite.
To be more precise, they are defined by
\begin{equation}
\D^\pm_{ijk} := \scr{D}^{\pm} \cap \{\S_{ii} > \S_{jj} > \S_{kk} > -\tfrac{1}{3}\} 
    = \scr{M}^{\pm}_{ij} \cap \scr{M}^{\pm}_{jk}.
\end{equation}
The following two examples play a role in the subsequent analysis:
\begin{align*}
\D^-_{123} &= \left\{ (0,0,\pi,\S_+,\S_-)~|~ 
    \S_+ < - \tfrac{1}{\sqrt{3}} \S_- < 0 \text{~and~}
    \S_+ - \sqrt{3} \S_- > -1 \right\}, \\
\D^+_{132} &= \left\{ (0,0,0,\S_+,\S_-)~|~ 
    \S_+ < \tfrac{1}{\sqrt{3}} \S_- < 0 \text{~and~}
    \S_+ + \sqrt{3} \S_- > -1 \right\}. 
\end{align*}
In Figure~\ref{fig:Dijk} the projections of the sets $\D_{123}^\pm$
and $\D_{132}^\pm$ onto the Jacobs disc $\scr{D}$,
i.e. in the $\S_+ \S_-$-plane, are shown.


%
%
\section{Monotonicity of the matter density} \label{sec:matterdensity}
\noindent 
Considering Equation \eqref{eq:EvEOm} it is clear that $\Omega$ is monotonically decreasing.
Since $ \Omega \in [0,1]$, as a consequence of the Hamiltonian constraint \eqref{eq:HCNew}, 
we conclude by the monotonicity principle, i.e. Lemma~\ref{lemma:monotonicity},
that if $\Omega(x) > 0$ then $\Omega(\tau)$ must converge to a positive value
$\Omega_0 \in [\Omega(x), 1]$ as $\tau \to \pm \infty$.
Hence $2-q \in \Lp{1}$, or equivalently $M \in \Lp{2}$.
As $M$ is a smooth function of time,
as it solves a smooth (in fact analytic) differential equation,
and as $M'$ is bounded due to the constraint \eqref{eq:HCNew},
we may conclude that $M \to 0$ as $\tau \to -\infty$.
On the other hand, as~$\tau \to \infty$, either we have that $\Omega(\tau) \to 0$,
or again $M \in \Lp{2}$.
In particular, we have proven the following.
\begin{lemma} \label{lemma:IntegrabilityOfM}
Let $x \in \bpvio$ for $\gamma = 2$.
Then towards the past $\Omega(\tau)$ converges to a value larger or equal than $\Omega(x)$,
$2-q \in \Lp{1}$,
$M \in \Lp{2}$ and $M \to 0$ as $\tau \to -\infty$.
\end{lemma}

\begin{rmk} \label{rmk:decelaration}
Observe that we may also write the deceleration parameter in the stiff fluid case as 
$q = 2- 2 N_-^2 + \tfrac{2}{3} A^2 = 2 + 3 \roScal_h/\theta^2$,
where $\roScal_h$ denotes the spatial scalar curvature.
\end{rmk}

\noindent Knowing that $M \in L^2((-\infty,0])$, we obtain several immediate consequences.
To begin we have the following.

\begin{lemma} \label{lemma:IntegrabilityOfP}
Let $x \in \bpvio$ for $\gamma = 2$.
Then $\S_+(\tau)$ converges as $\tau \to -\infty$,
$P~\in~L^2((-\infty,0])$ and $P(\tau) \to 0$ as $\tau \to -\infty$.
\end{lemma}
\begin{proof}
By Remark~\ref{rmk:alphalimit} we know that $\alpha(x) \neq \emptyset$.
Thus, let $y \in \alpha(x)$
and let $(\tau_k)_{k=1}^\infty$
be a decreasing sequence of times such that $\tau_k \to -\infty$ and $\vp^{\tau_k} (x) \to y$.
From the evolution equation \eqref{eq:EvESplus} we may deduce that
\begin{equation*}
3 \int_{\tau_k}^0 P(s)^2 d s = \S_+(0) - \S_+(\tau_k)
    + \int_{\tau_k}^0 \left[ (2-q(s)) \S_+(s) + 2 M(s)^2 \right] ds.
\end{equation*}
The integrands on the right-hand side are all integrable on the interval $(-\infty,0]$
by Lemma~\ref{lemma:IntegrabilityOfM}, the fact that $\S_+$ is bounded
and the fact that $\S_+(\tau_k)$ converges to $\S_+(y)$.
Hence we conclude that $\lim_{k \to \infty} \int_{\tau_k}^0 P(s)^2 ds$ exists
and thus $P \in L^2((-\infty,0])$.
As $P$ is smooth and has bounded derivative, it must converge to 0.
By rewriting the above equation, we may conclude that also $\S_+(\tau)$ must converge
as $\tau \to -\infty$.
\end{proof}
\begin{rmk} \label{rmk:DefinitionOfX}
In the analysis that follows
the functions $\sin(\theta)$ and $X := \sqrt{3} \S_- \cos(\theta)$
play a fundamental role.
For that reason, let us explicitly consider their evolution equations.
We have:
\begin{subequations} \label{eq:EveXSinTheta}
\begin{align}
\sin(\theta)' &= - 2 X \sin(\theta), \label{eq:EvESinTheta}\\
X' &= -(2-q)X + 2 \sin^2(\theta) \big( 3 \S_-^2 + \cos^2(\theta) \S_+^2\big)
    - 3 \cos^2(\theta) \big(2 M^2 + P^2 \big). \label{eq:EveX}
\end{align}
\end{subequations}    
(In particular, $\sin(\theta)$ may be integrated from $X$.)
\end{rmk}
The following proposition,
which is based on a rather simple analysis of the evolution equation for $X$,
is fundamental to the understanding of the asymptotics,
as we obtain a trichotomy of possible past limits. 
Firstly, we obtain that all variables of the solution converge.
Secondly, one of the following three things may happen: the backward converges to $\scr{F}$,
$\sin(\theta) \to 0$ and the shear diagonalizes in the limit towards the past,
or $\cos(\theta) \to 0$ and asymptotically the solution 
has the same behaviour as that a polarized solution.
\newpage
\begin{prop} \label{prop:ConvergenceOfSigma}
Let $x \in \bpvio$ for $\ga= 2$.
Then $\limtau \vp^\tau(x)$ exists and $\limtau X(\tau) \leq 0$. 
Moreover, 
we have that either
\begin{enumerate}[label=(\alph*)]
    \item 
    $\limtau \vp^{\tau}(x) \in (\scr{D}^+ \cup \scr{D}^-) \setminus \scr{F}$, or
    \item 
    $\limtau \vp^{\tau}(x) \in \scr{P} \setminus \scr{F}$, or
    \item 
    $\limtau \vp^{\tau}(x) \in \scr{F}$.
\end{enumerate}
\end{prop}
\begin{proof}
If $x \in \spvio$, then $\S_-(\tau) = 0$ and $\cos(\theta(\tau)) = 0$
throughout the evolution, and as $M(\tau) \to 0$ and $P(\tau) \to 0$,
and $\S_+(\tau)$ converges by Lemmata~\ref{lemma:IntegrabilityOfM} and~\ref{lemma:IntegrabilityOfP}, 
the conclusions stated in $(ii)$ apply, so we may consider $x \notin \spvio$.
In particular, as neither $\{\S_- = 0\}$ nor $\{\cos(\theta) = 0\}$ are invariant sets,
we may assume that $X (\tau) \neq 0$ for some $\tau \in (-\infty,0]$,
for else $\sin(\theta)$ grows exponentially toward the past.

Observe that, as $X$ is bounded and $2-q, M^2, P^2 \in L^1((-\infty,0]$
by Lemmata~\ref{lemma:IntegrabilityOfM} and~\ref{lemma:IntegrabilityOfP},
the first and last terms on the right-hand side of the evolution equation \eqref{eq:EveX} 
are integrable on the interval $(-\infty,0]$.
Since the term in the middle has a sign, we can apply a similar argument 
    as in Lemma~\ref{lemma:IntegrabilityOfP} to conclude that the function
\begin{equation*}
g(\tau):=\sin^2(\theta(\tau)) \big(3\S_-(\tau)^2
    + \cos^2(\theta(\tau)) \S_+(\tau)^2 \big) 
\end{equation*}
must be in $\Lp{1}$, and that the function $X$ has a limit towards the past.
In particular, $g(\tau)$ must converge to 0 as $\tau \to -\infty$,
as it is continuous, positive, integrable and has bounded derivative.
There are thus two possibilities as $\tau \to -\infty$, either $\sin(\theta(\tau)) \to 0$,
or $\cos(\theta(\tau)) \S_+(\tau) \to 0$ and $\S_-(\tau) \to 0$. 

Indeed, as $X(\tau)$ converges, we may integrate $\sin(\theta(\tau))$
and deduce that $\sin(\theta(\tau))$ converges as $\tau \to -\infty$,
and we must also have that $\limtau X(\tau) \leq 0$.
It follows readily that also $\theta(\tau)$ converges by connectedness of $\alpha(x)$,
and from the convergence of $X(\tau)$ the convergence of $\S_-(\tau)$ follows,
as long as $\cos(\theta(\tau))$ does not converge to zero.
However, if $\cos(\theta(\tau)) \to 0$, then as $g(\tau) \to 0$,
we obtain that $\S_-(\tau) \to 0$.
On the other hand, $\S_+(\tau)$ converges by Lemma~\ref{lemma:IntegrabilityOfP},
so firstly we have shown that $\limtau \vp^\tau(x)$ exists,
and secondly $3\S_-(\tau)^2 + \cos^2(\theta(\tau)) \S_+(\tau)^2$ converges,
    and if it does not convergence to zero,
then we may conclude from the fact that $g(\tau) \to 0$ that $\sin(\theta(\tau)) \to 0$.
If, on the other hand, it does converge to zero, then $\S_-(\tau) \to 0$
and $\cos(\theta(\tau)) \S_+(\tau) \to 0$ as $\tau \to -\infty$.

If $\limtau \S_+ \neq 0$, 
we may conclude that either $\sin(\theta(\tau)) \to 0$,
in which case $(i)$ holds, i.e. $\limtau \vp^{\tau}(x) \in (\scr{D}^+ \cup \scr{D}^-) \setminus \scr{F}$,
or that $\S_-(\tau) \to 0$ and $\cos(\theta) \to 0$, 
in which case $(ii)$ holds, i.e. $\limtau \vp^{\tau}(x) \in \scr{P} \setminus \scr{F}$.

On the other hand, if $\limtau \S_+(\tau) = 0$,
then we only know that either $\S_-(\tau) \to 0$ or that $\sin(\theta(\tau)) \to 0$.
In the latter case we find that, depending on the limit of $\S_-$, either $(i$) holds, i.e.
$\limtau \vp^{\tau}(x) \in (\scr{D}^+ \cup \scr{D}^-) \setminus \scr{F}$,
or $(iii)$ holds, i.e. $\limtau \vp^{\tau}(x) \in \scr{F}$.
In the former case we have that $(iii)$ holds, i.e. $\limtau \vp^{\tau}(x) \in \scr{F}$.
\end{proof}

\begin{notation} \label{not:Limits}
By the content of Proposition~\ref{prop:ConvergenceOfSigma},
we may for a given point in $\bpvi$ consider its limit and we denote 
by $\psi, s_+, s_-$ and $\mu$ the limits of $\theta, \S_+, \S_-, \Omega$ respectively.
Observe that $\mu$ may be computed from the other limits as
\begin{equation*}
   1 - \mu = \big(1 + \tfrac{1}{3}\sin^2(\psi) \big) s_+^2 + s_-^2.
\end{equation*}
We moreover use the notation
$s_+ = s_+(x) = \limtau \S_+(\vp^\tau(x))$.
Lastly, we denote $\xi:=\limtau X(\tau) = \sqrt{3} \cos(\psi) s_-$
for convenience.
\end{notation}

\begin{rmk} \label{rmk:ConvergenceOfSigma}
Recall from Remark~\ref{rmk:ReconstructionOfSigma} that we may reconstruct 
the components of the expansion-normalized shear with respect to a canonical basis from the variables 
$M, P, \theta, \S_+, \S_-$.
By Proposition~\ref{prop:ConvergenceOfSigma} we may then conclude
that the all the above functions converge.
In case $(i)$ all the off-diagonal components converge to zero as $\tau \to -\infty$,
while in case $(ii)$ we find that $\S_{11} + 2 \S_{23} \to 0$ and similarly $\S_{22} + \S_{23} \to 0$
and $\S_{33} + \S_{23} \to 0$,
and lastly in case $(iii)$ all components vanish in the limits.
\end{rmk}

Based on the trichotomy that is obtained in Proposition~\ref{prop:ConvergenceOfSigma},
we may define the following subsets of the non-vacuum phase space $\bpvio$
for $\gamma =2$,
namely we have the set of \emph{asymptotically-diagonalizing} orbits, denoted by $\advio$,
the set of \emph{half-polarized} orbits, denoted by $\hpvio$,
and the set of (past) \emph{isotropizing} orbits, denoted by $\isvio$.
These are defined respectively by cases $(i)$, $(ii)$ and $(iii)$ of Proposition~\ref{prop:ConvergenceOfSigma} above, namely
\begin{align}
\advio :=& \{ x \in \bpvio ~|~ \sin(\psi(x)) = 0 ~\text{and either}~ s_+(x) \neq 0 ~\text{or}~ s_-(x) \neq 0 \} \\
    =& \{ x \in \bpvio ~|~ \limtau \vp^\tau(x) \in
    (\scr{D}^+ \cup \scr{D}^-) \setminus \scr{F} \}, \nonumber \\
\hpvio :=& \{ x \in \bpvio ~|~ \cos(\psi(x)) = 0 = s_-(x) ~\text{and}~ s_+(x) \neq 0 \}  \\
    =& \{x \in \bpvio ~|~ \limtau \vp^\tau(x) \in \scr{P} \setminus \scr{F} \}, 
        \nonumber \\
\isvio :=& \{ x \in \bpvio ~|~ s_+(x) = 0 = s_-(x) \} \\
    =& \{x \in \bpvio ~|~ \limtau \vp^\tau(x) \in \scr{F} \}
        \nonumber.
\end{align}
It follows from Proposition~\ref{prop:ConvergenceOfSigma}
that these are well-defined, invariant subsets of the $\bpvio$,
they are mutually disjoint and lastly their union is all of $\bpvio$.
Moreover, the set $\hpvio$ contains the orbits of $\spvio$ for which $s_+ \neq 0$.

Asymptotically-diagonalizing refers to the shear diagonalizing
with respect to a canonical frame.
Half-polarized refers to the fact that despite these orbits not all being
part of the set of polarized orbits, the asymptotical behaviour is similar,
cf. the comments in the introduction in Subsection~\ref{sec:OverviewOfResults}.
Isotropizing refers to the eigenvalues of the expansion-normalized Weingarten map (or shear for that matter)
all converging to the same value toward the initial singularity.

\begin{rmk}
We could similarly try to define subsets $\advi^0$ and $\hpvi^0$ 
of the vacuum phase space $\bpvi^0$
replacing $\scr{D}^+ \cup \scr{D}^-$ and $\scr{P}$ by their vacuum boundaries respectively.
However, these might not yield well-defined invariant sets,
as convergence is certainly not guaranteed and in fact conjectured not to hold in that setting generically,
cf. \cite{HewittHorwoodWainrightExceptional}.
\end{rmk}

The sets $\advio$ and $\hpvio$ are investigated in the following two sections.
We wish to highlight the following two results.
On the one hand, in Section~\ref{sec:asympdiag},
in particular in Lemma~\ref{lemma:ExceptionalOrbitsNegativeSplus}, 
    it is shown that the set $\isvio$ of isotropizing orbits is contained within the 
    invariant set $\otpo$ of orthogonally-transitive orbits.
On the other hand, in Section~\ref{sec:halfpolarized} it is demonstrated
that the set of half-polarized solutions $\hpvio$
is a set of zero measure in $\bpvio$, 
in fact contained in a countable union of co-dimension~1 submanifolds.
In this sense the generic behaviour is that of orbits in $\advio$,
which is formalized and proven in Section~\ref{sec:MainTheorem}.
%
%
\section{The asymptotically-diagonalizing orbits} \label{sec:asympdiag}
\noindent
We consider first in more detail the set
    of asymptotically-diagonalizing orbits $\advio$. 
In this section we show in particular that for such orbits
the limit of the eigenvalues of the expansion-normalized Weingarten map are all positive,
or in other words, the solutions end up in the interior of the Jacobs set.
Moreover, orbits in this set end up in a rather specific region of the Jacobs set, 
depending on the index of the direction orthogonal to the Abelian $G_2$.

To begin, the following lemma shows that there are no orbits converging 
to the lines in $\scr{D}^\pm$ lying on the boundaries
of $\scr{M}_{23}^\pm$ and $\scr{M}_{32}^\pm$,
corresponding to the vanishing of one of the eigenvalues 
of the expansion-normalized Weingarten map.
\begin{lemma} \label{lemma:ExponentialM}
Let $x \in \advio$ for $\ga= 2$.
Then $s_+(x) + \xi(x) > -1$. 
In particular, $M(\tau) \to 0$ exponentially and also $q \to 2$ exponentially as $\tau \to -\infty$.
\end{lemma}
\begin{proof}
From Lemma~\ref{lemma:IntegrabilityOfM} and Proposition~\ref{prop:ConvergenceOfSigma}
we know that $M(\tau) \to 0$ as $\tau \to -\infty$
and we know that $\S_+, \S_-$ and $\theta$ converge.
We then may conclude that $s_+(x) +  \xi(x) \geq -1$,
because else $M$ would grow exponentially, 
as $\log(M)' = q + 2 \S_+ + 2 X \to 2 + 2 s_+ + 2 \xi$.
What remains to prove is that $s_+(x) + \xi(x) = -1$ cannot happen,
thus assume to the contrary that $x \in \advio$ is such that $s_+(x) + \xi(x) = -1$.

Firstly, by Lemma~\ref{lemma:IntegrabilityOfM} it holds that $\mu(x) >0$
and thus, by taking the limit of \eqref{eq:HCNew}, we find that 
\begin{equation*}
s_+^2(x) + s_-^2(x) = 1 - \mu(x) < 1.
\end{equation*}
It follows that $s_+(x) > -1$ and thus $\xi(x) < 0$.
Now, as $\sin(\theta)' = - 2 X \sin(\theta)$,
so that 
\begin{equation*}
\log(\sin(\theta))' \to - 2 \xi(x) > 0,
\end{equation*}
we find that $\sin(\theta(\tau)) \to 0$ exponentially as $\tau \to -\infty$.

On the other hand, consider $Y_M := 2 + 2\S_+ + 2 X$.
This function satisfies the evolution equation
\begin{equation} \label{eq:EvEYM}
Y_M' = -(2-q)Y_M + \sin^2(\theta) \big( 6 P^2 + 12 \S_-^2 
        + 4 \cos^2(\theta) \S_+^2 + \tfrac{40}{3} M^2 \big) - 12 M^2.
\end{equation}
We may now apply Lemma~\ref{lemma:EdgeCases} for the case that
$u,v,w, \alpha,\beta,m$ and $c$ are as in the table below.
\begin{center}
\begin{tabular}{c|c|c|c|c|c|c}
   $u$  & $v$ & $w$ &$\alpha$ & $\beta$ & $m$ & $c$  \\ \hline
   $M$  & $Y_M$ & 0 & $2-q$ & $2-q$ & 2 & 12 
\end{tabular}
\end{center}
Moreover, $f$ is the middle term Equation~\eqref{eq:EvEYM} above.
As $\sin^2(\theta(\tau)) \to 0$ exponentially as $\tau \to -\infty$ 
and the other variables are bounded, 
we obtain that $f$ is bounded by a function that
converges to zero exponentially as $\tau \to -\infty$.
We conclude that $\limtau Y_M (\tau) > 0$, thus we have a contradiction
to the assumption that $s_+(x) + \xi(x) = -1$.
Moreover, recalling Equation \eqref{eq:EvEM}, we find that
\begin{equation*}
Y_M = \log(M)' + 2-q,
\end{equation*}
we find that $\limtau \log (M(\tau))' > 0$ given that $2 - q(\tau) \to 0$ as $\tau \to -\infty$,
and thus $M(\tau) \to 0$ exponentially as $\tau \to -\infty$.
\end{proof}

Next, for the exceptional orbits, i.e. for orbits in $\bpvio \setminus \otpo$,
we can obtain more information as we know that also $P(\tau)$ must converge as $\tau \to -\infty$.
Observe that this statement also holds for the half-polarized orbits in $\hpvio$.

\begin{lemma} \label{lemma:ExceptionalOrbitsNegativeSplus}
Let $x \in \bpvio \setminus \otpo$ for $\gamma = 2$. 
Then $s_+(x) < 0$.
\end{lemma}
\begin{proof}
Let us consider the function $Z:= P^2 \sin(\theta)$.
As $x \notin \otpo$, $Z(x) > 0$ and we may compute that
\begin{align}
Z'  &= 2\big[-(2-q) - (3 + \sin^2(\theta)) \S_+ \big] Z.
\end{align}
Then, as $P^2 \geq Z$, we have that
\begin{align*}
-6\S_+' &= -(2-q) (-6\S_+) + 12 M^2  - 18 P^2 \\
    &\leq -(2-q)(-6 \S_+) + 12 M^2 - 18 Z.
\end{align*}
Thus we may now apply Lemma~\ref{lemma:EdgeCases}
for the case that $u,v,w, \alpha,\beta,m$ and $c$ are are as in the table below.
\begin{center}
\begin{tabular}{c|c|c|c|c|c|c}
   $u$  & $v$ & $w$ &$\alpha$ & $\beta$ & $m$ & $c$  \\ \hline
   $Z$  & $-6\S_+$ & $\tfrac{1}{3} \sin^2(\theta)$ & $2(2-q)$ & $2-q$ & 1 & 18 
\end{tabular}
\end{center}
Moreover $f = 12 M^2$,
where we note that $M(\tau)^2 \to 0$ exponentially as $\tau \to -\infty$
by Lemma~\ref{lemma:ExponentialM}.
We may thus conclude that $s_+(x) < 0$.
\end{proof}

Next we can show that unless $s_+(x) = 0$,
which by the above can only happen if $x \in \otpo$,
we must have that $\xi(x) < 0$ for any $x \in \advio$.

\begin{lemma} \label{lemma:ExponentialSinTheta}
Let $x \in \advio$ for $\ga= 2$.
If $s_+(x) \neq 0$ then $\xi(x) < 0$.
In particular, $\sin(\theta(\tau)) \to 0$ exponentially as $\tau \to -\infty$.
\end{lemma}
\begin{proof}
It must be so that $\xi(x) \leq 0$ because of the equation 
    $\log(\sin(\theta))' = - 2 X$.
As $\cos^2(\psi(x)) = 1$, we only have to show that $s_-(x) = 0$ is not possible.
Therefore assume to the contrary that $x$ is such that $s_-(x) = 0$, which implies that also $\xi(x) = 0$.

If $x \notin \otpo$, i.e. $P(x) > 0$, then as 
\begin{equation}
\log(P)' = -(2-q) - (3 + \sin^2(\theta)) \S_+ + X,
\end{equation}
we obtain that $\limtau \log(P(\tau))' = -3 s_+(x)$.
On the other hand, by Lemma~\ref{lemma:ExceptionalOrbitsNegativeSplus}
we must in fact have that $s_+(x) < 0$.
Therefore $P(\tau) \to 0$ exponentially as $\tau \to -\infty$.

On the other hand, if $P(x) = 0$,
it remains so as $\otpo$ is an invariant set.
Note that also $M(\tau) \to 0$ exponentially as $\tau \to -\infty$ by
Lemma~\ref{lemma:ExponentialM},
so that for any $x \in \advio$ we have that $M(\tau)$ and $P(\tau)$ may be bounded by a function
that converges to zero exponentially as $\tau \to -\infty$.

Let $T < 0$ be such that for any $\tau \leq T$
it holds that $\cos^2(\theta) > \frac{1}{2}$, 
and such that $|\S_+^2 - s_+^2| < \frac{1}{2} s_+^2$.
Such a $T$ exists as $\cos^2(\theta) \to 1$ and $\S_+^2 \to s_+^2 >0$.
We may rewrite and upper bound the evolution equation for $X$, 
i.e. Equation \eqref{eq:EveX}, as
\begin{align*}
- 2X' &= -(2-q)(-2X) 
    - 12 \sin^2(\theta) \S_-^2
    + 6 \cos^2(\theta)(2M^2 +P^2) \\
    &\quad - 4 \sin^2(\theta) [\cos^2(\theta) - \tfrac{1}{2} + \tfrac{1}{2}] 
        (\S_+^2 - s_+^2 + s_+^2) \\
    &\leq -(2-q)(-2X) 
    + 6 \cos^2(\theta)(2M^2 +P^2)
    - \sin^2(\theta) s_+^2
\end{align*}
Now we apply Lemma~\ref{lemma:EdgeCases}
for the case that
$u,v,w, \alpha,\beta,m$ and $c$ are as in the table below.
\begin{center}
\begin{tabular}{c|c|c|c|c|c|c}
   $u$  & $v$ & $w$ &$\alpha$ & $\beta$ & $m$ & $c$  \\ \hline
   $\sin(\theta)$  & $-2X$ & 0 & 0 & $2-q$ & 2 & $s_+^2$ 
\end{tabular}
\end{center}
Here the relevant interval is $(-\infty, T]$,
where we note that $f = \cos^2(\theta)(M^2 + P^2)$ may be bounded
by a function exponentially converging to zero towards the past.
The conclusion is that $\xi(x) < 0$, which is a contradiction to $s_-(x) =0$.
\end{proof}

\begin{lemma} \label{lemma:ExponentialP}
Let $x \in \advio \setminus \otpo$ for $\gamma = 2$.
Then $3 s_+(x) - \xi(x) < 0$.
In particular, $P(\tau) \to 0$ exponentially as $\tau \to -\infty$.
\end{lemma}
\begin{proof}
From Lemma~\ref{lemma:IntegrabilityOfP} we know that $P(\tau) \to 0$ as $\tau \to -\infty$.
We then may conclude that $3s_+(x) - \xi(x) \leq 0$,
because else $\limtau \log(P)' < 0$, thus $P$ would grow exponentially.

Now assume to the contrary that $3 s_+(x) = \xi(x) $
for some $x \in \advio \setminus \otpo$.
From Lemma~\ref{lemma:ExceptionalOrbitsNegativeSplus} we know that $s_+(x) < 0$,
and we may conclude that $\xi(x) < 0$, and in particular 
we have that $\sin(\theta(\tau)) \to 0$ exponentially as $\tau \to -\infty$.

Observe on the other hand, that the function
$Y_P:= -(3+\sin^2(\theta)) \S_+ + X$,
satisfies
\begin{align*}
Y_P' &= - (2-q) Y_P + 8 \sin^2(\theta) M^2 + 2 \sin^2(\theta)
    \big( \sqrt{3} \S_- + \cos(\theta) \S_+ \big)^2 - 12 P^2.
\end{align*}
Now we may apply Lemma~\ref{lemma:EdgeCases}
for the case that
$u,v,w, \alpha,\beta,m$ and $c$ are as in the table below.
\begin{center}
\begin{tabular}{c|c|c|c|c|c|c}
   $u$  & $v$ & $w$ &$\alpha$ & $\beta$ & $m$ & $c$  \\ \hline
   $P$  & $Y_P$ & 0 & $2-q$ & $2-q$ & 2 & 12
\end{tabular}
\end{center}
Moreover $f$ is the sum of the two terms in the middle of
the right-hand side of the evolution equation above.
Observe that $f$ is bounded by a function that converges to 0 exponentially,
as $\sin^2(\theta)$ and $M^2$ converge to 0 exponentially as $\tau \to -\infty$.
The conclusion then reads that $3 s_+(x) < \xi(x)$, 
which yields a contradiction.
The exponential rate is then a simple consequence of the fact that 
\begin{equation*}
Y_P = \log(P)' + 2 - q,
\end{equation*}
thus $\limtau \log(M)' >0$,
given that $2-q \to 0$ as $\tau \to -\infty$,
and thus $P(\tau)$ converges to zero exponentially as $\tau \to -\infty$.
\end{proof}

\begin{rmk} \label{rmk:EigenvaluesOfENSFF}
From Remark~\ref{rmk:ConvergenceOfSigma}, we know that with respect to a canonical frame
the components of the expansion-normalized shear $\S$ converge.
Denoting the expansion-normalized Weingarten map by $K$,
then also the components of $K$ converge with respect to a canonical frame.
We denote the limits of the components by $\rK_{ij}$, for $i,j \in \{1,2,3\}$.

For an orbit in $\advio$,
all off-diagonal elements of the expansion-normalized shear converge to~0 toward the past.
Hence the limits of the eigenvalues of the expansion-normalized Weingarten map
may be read off the diagonal of $\rK$ written in canonical coordinates.

Thus they can be reconstructed from $s_{+}$ and $s_{-}$ as
\begin{align*}
\rK_{11} &= \tfrac{1}{3} \big( 1 - 2 s_+ \big), \\
\rK_{22} &= \tfrac{1}{3} \big( 1 + s_{+} + \sqrt{3} s_- \big), \\
\rK_{33} &= \tfrac{1}{3} \big( 1 + s_{+} - \sqrt{3} s_- \big).
\end{align*}
Consider, for example, $x \in \advio$ for $\gamma = 2$,
for which $\psi(x) = 0$.
Then, as $q \to 2$, it follows that
$\rK_{22}$ is equal to the past limit of $\frac{1}{6}\log(M)'$,
while $\rK_{33}$ is the past limit of $\frac{1}{6} \log(M \sin^2(\theta))'$
and $\rK_{11}$ is the past limit of $\frac{1}{6}\log(M P^2 \sin^2(\theta))'$.
\end{rmk}

\subsection{The orthogonally-transitive case}
Let us summarize the information we obtained about the orthogonally-transitive case,
i.e. regarding $x \in \advio \cap \otpo$.
The results we obtain in this case are not entirely new,
cf. \cite{HewittWainwrightClassB} and \cite{RadermacherStiff},
although they are slightly more refined.
Note that $\advio$ excludes solutions converging to the set $\scr{F}$.
We obtain that $s_+(x) + \xi(x) > -1$ by Lemma~\ref{lemma:ExponentialM}.
On the other hand, if $s_+(x) = 0$, then $s_-(x) \neq 0$, for else $\alpha(x) \subset \scr{F}$,
and thus by Proposition~\ref{prop:ConvergenceOfSigma} we must have $\xi(x) < 0$.
But also if $s_+(x) \neq 0$ then $\xi(x) < 0$ by Lemma~\ref{lemma:ExponentialSinTheta}.
The intersection of these two conditions means that $\vp^{\tau}(x)$
converges to a point in $\scr{M}^+_{32} \cup \scr{M}^-_{23}$ as $\tau \to -\infty$.
We refer to Figure~\ref{fig:Mij} for 
the projections of these sets to the Jacobs disc $\scr{D}$ in the $\S_+ \S_-$-plane.

For orbits in $\advio \cap \otpo$ the limits of eigenvalues of the expansion-normalized Weingarten map are not necessarily all positive.
If the backward orbit converges to a point such that $\S_+ \geq 1/2$,
then we recall that then $\rK_{11} \leq 0$,
cf. Remark~\ref{rmk:EigenvaluesOfENSFF} above.
In canonical coordinates, cf. Appendix~\ref{sec:AppendixEquations},
$e_1$ corresponds to the direction $a^\#$,
which is the direction orthogonal to the Abelian $G_2$,
and in the orthogonally-transitive case is also an eigenvector of the shear,
thus also of the expansion-normalized Weingarten map $K$.
Note moreover that the largest eigenvalue of $K$ is distinct from the others.

In summary, we have the following proposition.
\begin{prop} \label{prop:NonExceptionalAD}
Let $x \in \advio \cap \otpo$ for $\gamma = 2$.
Then $\vp^{\tau}(x)$ converges to a point in $\scr{M}_{23}^- \cup \scr{M}_{32}^+$
as $\tau \to -\infty$.
In particular, the limits of the eigenvalues of the expansion-normalized Weingarten map
exist and if one of them is non-positive, then it is the limit
of the eigenvalue whose eigenspace coincides with the direction orthogonal
to the Abelian $G_2$.
\end{prop}

\subsection{The exceptional case}
Of greater interest to us are the exceptional orbits.
For orbits in $\advio \setminus \otpo$, it holds
    that $s_+(x) + \xi(x) > -1$ by Lemma~\ref{lemma:ExponentialM},
    that $s_+(x) < 0$ by Lemma~\ref{lemma:ExceptionalOrbitsNegativeSplus},
    that $\xi(x) < 0$ by Lemma~\ref{lemma:ExponentialSinTheta},
    and that $3s_+(x) - \xi(x) < 0$ by Lemma~\ref{lemma:ExponentialP},
The intersection of these conditions means that
$\vp^{\tau}(x)$ converges to a point in $\D^-_{123} \cup \D^+_{132}$ 
    as $\tau \to -\infty$.
We refer to Figure~\ref{fig:Dijk} for 
the projections of these sets to the Jacobs disc $\scr{D}$ in the $\S_+ \S_-$-plane.
In particular, the limits $\rK_{11}, \rK_{22}, \rK_{33}$,
which coincide with the limits of the eigenvalues 
of the expansion-normalized Weingarten map, are distinct and positive.
In summary, we have the following proposition.
\begin{prop} \label{prop:ExceptionalAD}
Let $x \in \advio \setminus \otpo$ for $\gamma = 2$.
Then $\vp^{\tau}(x)$ converges to a point in $\D_{123}^- \cup \D_{132}^+$
as $\tau \to -\infty$.
In particular, the limits of the eigenvalues of the expansion-normalized Weingarten map exist,
are distinct and strictly positive.
\end{prop}

\begin{rmk} \label{rmk:NoAsysmptoticEigenspaces}
We know that the off-diagonal components $K_{12}$ and $K_{13}$ of the expansion-normalized Weingarten map
converge to zero towards the past.
This fact alone is not enough, as far as we are aware, 
to conclude that the eigenspace that corresponds to the largest eigenvalue of $K$
converges to the direction orthogonal to the $G_2$.
This is a subtle issue, due to the anisotropy and the lack of knowledge
regarding the past convergence of the actual vectors that lie either orthogonal
or tangent to the $G_2$.
\end{rmk}


%
%
\section{The half-polarized orbits} \label{sec:halfpolarized}
\noindent
Recall that we defined the half-polarized orbits $\hpvio$
as the set of non-vacuum orbits for which $\cos(\psi) = 0$,
or, alternatively,
which towards the past converge to a point in $\scr{P}$.
In particular, this contains the set $\spvio$ which is of co-dimension~2.
In this section we show that the set $\hpvio$ is
a null subset of the phase space $\bpvio$ of non-vacuum orbits for $\gamma = 2$.

\begin{rmk} \label{rmk:HalfPolarizedEigenvaluesWeingartenMap}
For the exceptional asymptotically-diagonalizing orbits,
i.e. orbits of points in $\advio \setminus \otpo$,
we saw in Proposition~\ref{prop:ExceptionalAD} that the limits of the eigenvalues
    of the expansion-normalized Weingarten map $K$ are all positive.
This is not necessarily the case for the exceptional half-polarized orbits,
i.e. orbits of points in $\hpvio \setminus \otpo$.
Indeed, with respect to a canonical frame, by Remarks~\ref{rmk:ReconstructionOfSigma} and~\ref{rmk:ConvergenceOfSigma}
the limit of the matrix representing the expansion-normalized shear $\S$ is given by
\begin{equation*}
\frac{s_+}{3} \begin{pmatrix} 
-2 & 0 & 0 \\
0 & 1 & 1 \\
0 & 1 & 1
\end{pmatrix}.
\end{equation*}
Thus its eigenvalues are $-\frac{2}{3} s_+, \frac{2}{3} s_+$ and $0$.
The limit of the Hamiltonian constraint in this case is 
\begin{equation*}
\tfrac{4}{3} s_+^2 + \mu =1,    
\end{equation*}
so that we have $s_+ = -\frac{1}{2} \sqrt{3} \sqrt{1 - \mu}$,
recalling that $s_+(x) < 0$ for $x \in \hpvio \setminus \otpo$
by Lemma~\ref{lemma:ExceptionalOrbitsNegativeSplus}.
Hence we obtain that the smallest limit of the eigenvalues of $\S$ lies 
in the range $(-\frac{1}{\sqrt{3}}, 0)$,
so the smallest limit of the eigenvalue of $K$
lies in the range $(-\frac{1}{\sqrt{3}} + \frac{1}{3}, \frac{1}{3})$,
and there appears to be no obstruction to a negative eigenvalue.
\end{rmk}

\subsection{Non-genericity of the half-polarized solutions}
We wish to demonstrate that the set of orbits $\hpvio$ constitutes a null set in $\bpvio$.
For this reason we may consider the system \eqref{eq:PolarEvE}
and its linearization at the points in the set $\scr{P}$.
It will turn out to be convenient to split up the line segment $\scr{P}$ as 
\begin{equation}
\scr{P}^{-} := \scr{P} \cap \{\S_+ < 0\}, \quad \scr{P}^+ := \scr{P} \cap \{\S_+ > 0\}.    
\end{equation}

\begin{rmk}
Observe that, as we have rewritten the system \eqref{eq:PolarEvE}
to be independent of $\Omega$ and $\S_\times$, 
and in particular as $\Omega$ and $\S_{\times}$ are functions of $M,P, \theta, \S_+,\S_-$,
we do not need to take their evolution equations into account
when studying the eigenvalues of the linearized system.
Usually one separates the directions tangent to the constraints
and transverse to the constraints, the latter deemed \emph{unphysical},
cf. e.g. Subsection~2.3 of \cite{HewittHorwoodWainrightExceptional},
as only the relative stability on the manifold defined by the constraints
is to be taken into account.
However, in the formulation of the evolution equations \eqref{eq:PolarEvE}
the directions transverse to the constraint are precisely
    $\nabla \Omega$ and $\nabla \S_{\times}$,
thus if we study the linearization without $\Omega$ and $\S_{\times}$ we obtain 
the same relative stability.
\end{rmk}

The linearization of the system of equations \eqref{eq:PolarEvE},
around a point $y$ of the form 
\begin{equation*}
y = (0,0, \pi/2, s_+, 0) \in \scr{P},
\end{equation*}
(with $\mu > 0$, i.e. non-vacuum) is given by 
\begin{subequations} \label{eq:LinearizationL}
\begin{align}
    \tilde{M}' &= (2 + 2 s_+) \tilde{M}, \\
    \tilde{P}' &= -4 s_+ \tilde{P}, \\
    \tilde{\theta}' &=  -2\sqrt{3} \tilde{\S_-}, \\
    \tilde{\S}_+' &= 0, \\
    \tilde{\S}_-' &= -\tfrac{2}{\sqrt{3}} s_+^2 \tilde{\theta},
\end{align}
\end{subequations}
and thus we find in particular that at $y$ the system has 
eigenvalues and eigenvectors as displayed in Table~\ref{tab:eigenvalues}.
 
\begin{table}[b]
\begin{tabular}{c | c | c | c | c | c}
Eigenvalue
& $2 + 2s_+$
& $-4 s_+$
& $- 2s_+$
& $2 s_+$
& $0$ \\ \hline 
Eigenvector
& $\partial_M$
& $\partial_P$
& $\sqrt{3} \partial_{\theta} + s_+ \partial_{\S_-}$
& $\sqrt{3} \partial_{\theta} - s_+ \partial_{\S_-}$
& $\partial_{\S_+}$ 
\end{tabular} 
\caption{The eigenvalues and corresponding eigenvectors
of the matrix of the linearization at a point
$(0,0,\pi/2, s_+, 0) \in \scr{P}$.}
\label{tab:eigenvalues}
\end{table}
If $y \in \scr{P}^-$, 
there are thus three positive eigenvalues, one negative eigenvalue and one eigenvalue equals zero.
If on the other hand for $y \in \scr{P}^+$, i.e. $s_+(y) > 0$,
which only non-exceptional orbits converge to towards the past,
then we have two positive eigenvalues, two negative eigenvalues and again one eigenvalue equals zero.

The idea is now to use Theorem B.3 from Appendix B of \cite{RadermacherVacuum}
in order to obtain the existence of local center-unstable manifolds.\footnote{
The original reference is Theorem 9.1 of \cite{FenichelGeomSingPert}.
However, the formulation in \cite{RadermacherVacuum} suffices for our purposes.
In \cite{FenichelGeomSingPert} it is also outlined
how, in principle, asymptotic expansions at the fixed points may be obtained to all orders.}
Once we have the local center-unstable manifolds, we may then show that 
the unstable set to $\hpvio$, i.e. the set of points of orbits converging to $\scr{P}$,
is contained in a countable union of $C^{r}$-submanifolds of co-dimension~1,
where $r$ may be chosen to be any positive integer.

\begin{prop} \label{prop:NonGenericityOfHalfPolarized}
The set $\hpvio$ of half-polarized non-vacuum orbits, for $\gamma = 2$,
is contained in a countable union of $C^r$ submanifolds of $\bpvio$ of co-dimension~1,
where $r$ may be chosen to be any positive integer.

Moreover, there exists a $C^r$ submanifold of $\bpvio$ of co-dimension~1, where $r$ is as above,
consisting of half-polarized orbits for which the smallest limit of the eigenvalues of the 
expansion-normalized Weingarten map is strictly negative.
\end{prop}
\begin{proof}
We already know from Proposition~\ref{prop:ConvergenceOfSigma}
and Lemma~\ref{lemma:ExceptionalOrbitsNegativeSplus}
that if $x \in \hpvio \setminus \otpo$, then $\vp^{\tau}(x)$
    converges to a point in $\scr{P}^{-}$ as $\tau \to -\infty$.
Given that the set of orthogonally-transitive orbits $\otpo$ is itself
a smooth submanifold of co-dimension~1,
this shows that the set of orbits converging to $\scr{P}^{+}$
is contained in a smooth submanifold of co-dimension~1.
Thus we only need to focus on the orbits converging to $\scr{P}^-$,
i.e. the set $\hpvio \setminus \otpo$.

Fix $r \in \mathbb{N}$.
By Theorem B.3 in \cite{RadermacherVacuum},
for every compact subset $\scr{Y} \subset \scr{P}^-$ there exists a neighbourhood
$\scr{N_Y} \subset \cl (\bpvi) \setminus \cl(\bpvi^0)$ of $\scr{Y}$
and a four-dimensional $C^r$ submanifold $\scr{M_Y} \subset \scr{N_Y}$ containing~$\scr{Y}$,
such that at every $y \in \scr{Y}$ the submanifold $\scr{M_Y}$ is tangent to 
    the eigenspaces corresponding to non-positive eigenvalues
    of the matrix of the linearization \eqref{eq:LinearizationL} at $y$.
Moreover, if $x \in \scr{N_Y}$ then either $x \in \scr{M_Y}$ or 
the backward orbit through $x$ leaves $\scr{N_Y}$.

Next, let us show that the backwards basin of attraction, denoted $\scr{U_Y}$,
of such a compact subset~$\scr{Y}$ is contained in a countable union of $C^r$ submanifold of $\bpvio$ of co-dimension~1.
All orbits in $\bpvio$, in particular those contained in $\scr{U_Y}$,
converge towards the past by Proposition~\ref{prop:ConvergenceOfSigma}.
Thus, for any $z \in \scr{U_Y}$ there exists a $T_z < 0$ such that
for $\tau \leq T_z$, it holds that $\vp^{\tau}(z) \in \scr{N_Y}$. 
Hence, by the conclusions above, for such $\tau$ the point $\vp^{\tau}(z)$ belongs to the manifold $\scr{M_Y}$.
Therefore, the family of $C^r$~co-dimension~1 submanifolds $\{\vp^{k}(\scr{M_Y})\}_{k=1}^{\infty}$ covers all of $\scr{U_Y}$  .

In particular, choosing a countable set of compact subsets $(\scr{Y}_m)_{m=1}^\infty$ 
    exhausting the set $\scr{P}^-$,
yields that the backwards basin of attraction of any $y \in \scr{P}^{-}$
is contained in a countable union of co-dimension~1 submanifolds of $\bpvio$.
We may thus conclude the same for the set $\hpvio$.

The second statement, i.e. the existence of the $C^r$ submanifold of $\bpvio$ of co-dimension~1, 
consisting of half-polarized orbits for which the smallest limit of the eigenvalues of the 
expansion-normalized Weingarten map is strictly negative,
is a consequence of applying Theorem B.3 to e.g. the following compact subset of $\scr{P}^-$, namely:
\begin{equation*}
\scr{Y}_\ve := \scr{P}^- \cap \big\{ 
    \S_+ \in \big[- \tfrac{1}{2}\sqrt{3} + \tfrac{3}{2} \ve,
        -\tfrac{1}{2}- \tfrac{3}{2}\ve \big] \big\}
\end{equation*}
for suitably small $\ve >0$. 
Then the smallest of the limits of the eigenvalues 
of the expansion-normalized Weingarten map of orbits on the corresponding centre-unstable manifold $\scr{M}_{\scr{Y}_\ve}$
is in the range $[- (\frac{1}{\sqrt{3}} - \frac{1}{3}) + \ve, -\ve]$,
cf. Remark~\ref{rmk:HalfPolarizedEigenvaluesWeingartenMap}.
\end{proof}
\section{Generic past convergence to the Jacobs disc} \label{sec:MainTheorem}
\noindent
We are now ready to formulate our main result concerning the generic behaviour of solutions towards the past, which now follow almost directly from
the results proven in Sections \ref{sec:asympdiag} and \ref{sec:halfpolarized}.
\newpage
\begin{thm} \label{thm:GenericConvergenceToTriangle}
For $\gamma = 2$, any backward orbit in the invariant set $\advio \setminus \otpo$,
which is a subset of $\bpvio$ of full measure that 
moreover contains a countable intersection of dense and open subsets of $\bpvio$,
converges to a point in $\Delta_{123}^- \cup \Delta_{132}^+$.

In particular, for the orbits in $\advio \setminus \otpo$ the past limits
of the eigenvalues of the expansion-normalized Weingarten map exist,
are distinct and are all strictly positive.
\end{thm}
\begin{proof}
As a corollary of Proposition~\ref{prop:NonGenericityOfHalfPolarized}
and Sard's theorem, cf. e.g. Corollary 6.12 of \cite{LeeSmoothManifolds},
$\advio \cup \isvio$ is a full measure subset $\bpvio$
that is a countable intersection of dense and open sets.
The same conclusions hold if remove $\otpo$ as this is a smooth submanifold of co-dimension~1.
Note also that the set $\isvio$ of isotropizing orbits is contained in $\otpo$ by 
Lemma~\ref{lemma:ExceptionalOrbitsNegativeSplus}.
The statement regarding convergence of the backward orbits 
then follows directly from Proposition~\ref{prop:ExceptionalAD}.
\end{proof}

\section{Towards the future} \label{sec:future}
\noindent So far the main focus has been on the behaviour of solutions towards the past.
However, using methods similar to those found in \cite{HewittTraceClass},
which concerns polarized type~\vi{} orthogonal perfect fluid Bianchi cosmologies,
we may obtain at least some bounds for certain quantities towards the future.
In this section we do not fix the parameter $\gamma$ at the value 2, 
but use Equations (\ref{eq:PolarEvEAppendix}) - (\ref{eq:AuxiliaryAppendix}) for more general $\gamma \in (2/3,2]$,
as the analysis below is of relevance also for non-stiff perfect fluids.

Central to this analysis are the functions $V_{\alpha, \beta, \delta, \varepsilon}$
defined as 
\begin{align}
V_{\alpha, \beta, \delta,\varepsilon} := \alpha \log( M ) + \beta \log(P) 
    + \delta \log( \sin(\theta)) + \varepsilon \log(\Omega),
\end{align}
where $\alpha, \beta, \delta, \varepsilon \in \bb{R}$ are real parameters, typically integer valued.
Observe that for example the function $Z = P^2 \sin(\theta)$,
which appeared in a proof above, is the same as $\exp(V_{0,2,1,0})$ in this notation.
The usefulness of these functions is due to the property that
\begin{align}
V_{\alpha, \beta, \delta, \varepsilon}'
    = (\alpha + \beta + 2 \varepsilon) q - 2 \beta - 2 q^* \varepsilon
    + (2 \alpha+ \beta - 2 \delta) X
    + (2 \alpha - (3 + \sin^2(\theta)) \beta) \S_+,
\end{align}
where we recall $q^*(\gamma) = \frac{1}{2}(3\gamma -2)$.
On the other hand, we know that the functions $M, P, \sin(\theta)$ and $\Omega$
    are all strictly positive yet bounded from above for orbits not in $\otpo$,
hence if the derivative of $V_{\alpha, \beta, \delta, \varepsilon}$
may be bounded above or below from 0,
then we obtain asymptotic information about the functions $M, P, \sin(\theta)$ and $\Omega$.

In particular, consider the function $V_{4,2,5,-3}$, which is well-defined on $\bpvio \setminus \otpo$ 
and whose derivative satisfies 
\begin{equation}
V_{4,2,5,-3}' = 9 \gamma - 10 + 2 \cos^2(\theta) \S_+.
\end{equation}
We obtain in particular that the derivative is strictly positive on its domain for $\gamma \geq 4/3$.
As a consequence of the monotonicity principle, i.e. Lemma~\ref{lemma:monotonicity},
for $x \in \bpvio \setminus \otpo$ and $\gamma \geq 4/3$ it must hold that
\begin{align*}
\lim_{\tau \to -\infty}& M P \sin(\theta) (\tau) = 0, \qquad
\lim_{\tau \to \infty} \Omega (\tau) = 0.
\end{align*}
Indeed, by the monotonicity principle we must end up on the boundary of $\bpvio \setminus \otpo$,
which consists of the intersection of $\bpvi$ with either $\{ M = 0\}$, $\{P =0 \}$, $\{\sin(\theta) = 0\}$
or $\Omega = 0$, i.e. respectively the sets corresponding to Bianchi type I, the $G_2$ acting orthogonally-transitively, Bianchi type II or vacuum.
It may of course also end up in intersections of these boundaries,
but in the limit towards the past at least one of $M = 0$, $P = 0$, or  $\sin(\theta) = 0$ must hold,
while toward the future necessarily $\Omega = 0$ must hold.
In particular, we have proven the following proposition.
\begin{prop}
Let $x \in \bpvio \setminus \otpo$ for $\gamma \in [4/3, 2]$.
Then 
\begin{equation}
\limtau M(\tau) P (\tau) \sin(\theta(\tau)) = 0
\end{equation} 
and 
\begin{equation}
\lim_{\tau \to \infty} \Omega(\tau) = 0.
\end{equation}
\end{prop}
\section*{Acknowledgements}
\noindent The author sincerely thanks Hans Ringström for the suggestion of the topic,
and his careful reading as well as many comments and suggestions on drafts of the manuscript.
Many thanks also go to Claes Uggla for the stimulating discussions on the subject
(as well as on many other subjects) and the insightful suggestions
regarding adapted polar coordinates and references in the literature.
Lastly, thanks go to Phillipo Lappicy
for the stimulating discussion on the exceptional Bianchi spacetimes during CERS XIII.
This research was funded by the Swedish research council, dnr. 2017-03863 and 2022-03053.
\appendix
%
%

\section{Bianchi types and their geometrical features} \label{sec:abelianG2}
\noindent In this appendix we recall some facts concerning the Bianchi-Behr classification.
We also discuss special geometrical features that arise in Bianchi spacetimes
that play a role in this article,
in particular the property of orthogonal transitivity of the Abelian $G_2$
and the potential presence of hypersurface-orthogonal Killing vector fields.

\subsection{The Bianchi-Behr classification} \label{sec:Classification}
Let us briefly recall some essentials regarding the Bianchi-Behr classification,
in particular the defining properties of the Bianchi type~\vi{} algebras;
see e.g. Section~1.5.1 of \cite{DynSysCosmology} and the references therein
or otherwise \cite{EllisMacCallumHomCos} for a more complete overview.

A connected, three-dimensional Lie group $G$ may be classified by its Lie algebra $\g$,
based on properties of the structure constants of $\g$.
The Lie algebra $\g$ is said to be of class A
if it is unimodular,
which is equivalent to the linear map $\ha \in \g^*$ being trivial,
where $\ha$ given by 
\begin{equation}
\ha: X \mapsto \tfrac{1}{2} \tr(\ad_X).
\end{equation}
If the algebra is not unimodular then $\g$ is said to be of class B, such as, for example, type~\vi{}.
Next, given a basis $\he_1, \he_2, \he_3$ of $\g$, we may consider the structure constants 
$\hg_{ij}^k$ defined by $[\he_i, \he_j] = \hg_{ijk}^k \he_k$.\footnote{
We use hats, e.g. $\he_1$, here to denote objects in the Lie algebra,
so as to distinguish them for objects in spacetime.}
We may then define the quantities
\begin{equation}
\hn^{kl} := \epsilon^{ij(k} \hg_{ij}^{l)}, \quad \ha_{i} := \tfrac{1}{2} \hg_{ik}^k.
\end{equation}
Observe that the $\ha_i$ are the components of the map $\ha$ above,
in particular, if $\ho^1, \ho^2, \ho^3$ is a basis dual to $e_1, e_2, e_3$,
then $\ha = \ha_i \ho^i$.
The Lie algebra being of class B thus means that not all $\ha_i$ vanish.

Concerning the Bianchi types, we consider the matrix $\hn$ with components $\hn^{ij}$, 
which manifestly is a symmetric matrix.
In particular, by rotating the frame we may diagonalize the matrix $\hn$,
and it has real eigenvalues.
The Jacobi identity may be written as $\hn^{ij}\ha_j = 0$
and thus $\ha^\# = \ha_i \he_i$ is an eigenvector
of the matrix $\hn$ with components $\hn^{ij}$ and it has eigenvalue 0.

For class A, we obtain the Bianchi types as in Table~\ref{tab:typeA},
depending only on the relative signs of the eigenvalues of the matrix $\hn$.
For class B, there is always one eigenvalue zero, so we can consider the relative
signs of the eigenvalues of the remaining two,
which leads to the other Bianchi types in Table~\ref{tab:typeB}.
If both of these eigenvalues are non-zero,
i.e. in the case of Bianchi type VI$_\eta$ (when they have opposite signs),
and Bianchi type VII$_\eta$ (when they have equal signs),
then there is an additional parameter $\eta$ determined by the relation
\begin{equation} \label{eq:HawkingsConstraint}
\ha_i \ha_j = \tfrac{\eta}{2} \epsilon_{ikl} \epsilon_{mnq} \hn^{km} \hn^{lq},
\end{equation}
an observation that goes back to \cite{CollinsHawkingIsotropic}.

Thus, for Bianchi type~\vi{}, the matrix $\hn$ has two non-zero eigenvalues
of opposite sign, and if written with respect to a basis diagonalizing the matrix $\hn$
and such that $\hn^{11} =0$,
we have the algebraic relation
\begin{equation}
\ha_1^2 = -\tfrac{1}{9} \hn^{22} \hn^{33}.
\end{equation}

\begin{table}[ht]
\captionsetup{position=below}
\begin{subtable}[h]{0.4\textwidth}
\begin{center}
{\renewcommand{\arraystretch}{1.2}
\begin{tabular}{l c c c c}
Class & Type & $ \hn^{11} $ & $ \hn^{22} $ & $ \hn^{33} $\\ \hline \hline 
A & I & 0 & 0 & 0  \\ \hline
  & II & 0 & 0 & +  \\ \hline
  & VI$_0$ & 0  & + & -  \\ \hline
  & VII$_0$ & 0 & + & +  \\ \hline
  & VIII & + & + & -  \\ \hline
  & IX & + & + & +  
\end{tabular}}
\caption{The various Bianchi types of class A.}
\label{tab:typeA}
\end{center}
\end{subtable}
\hspace{2cm}
\begin{subtable}[h]{0.4\textwidth}
\begin{center}
{\renewcommand{\arraystretch}{1.2}
\begin{tabular}{l c c c}
Class & Type & $ \hn^{22} $ & $ \hn^{33} $ \\ \hline \hline 
B & IV  & 0 & + \\ \hline
 & V & 0 & 0 \\ \hline
 & VI$_\eta$  & + & -  \\ \hline
 & VII$_\eta$  & + & +
\end{tabular}}
\caption{The various Bianchi types of class B;
observe that type III is the same as type VI$_{-1}$
and is thus omitted.}
\label{tab:typeB}
\end{center}
\end{subtable}
\caption{The Bianchi types of class A (left), for which $\ha = 0$, and class B (right),
for which $\ha \neq 0$, typified by the different signs of the eigenvalues of the matrix $\hn$.
Observe that for type VI$_\eta$ and VII$_\eta$
there is also the relation~\eqref{eq:HawkingsConstraint}.}
\end{table}

\subsection{The Abelian $G_2$}
All the Lie algebras in the Bianchi classification, except those of types VIII or IX,
possess an Abelian Lie subalgebra of dimension 2,
which integrates to a two-dimensional Abelian subgroup of the corresponding simply-connected Lie group. 
For the lower types of class A, namely I, II, VI$_0$ and VII$_0$,
the presence of the Abelian Lie subalgebra is a simple consequence of the classification itself.
Indeed, if we diagonalize $\hn$ and let $\hn^{11}$ = 0, then due to the diagonalization
we may write $[\he_2, \he_3] = \hn^{11} \he_1$ = 0.
In particular the span of $\he_2$ and $\he_3$ is an Abelian Lie subalgebra
integrating to an Abelian Lie subgroup.
For the types of class B, namely IV, V, VI$_{\eta}$ and VII$_{\eta}$,
something similar holds.
In this case, if we let $\ha^\# = \ha_1 \he_1$ and diagonalize $\hn$ so that
then by the Jacobi identity we have $\hn_{11} = 0$,
and then again $[\he_2, \he_3] = \hn^{11} \he_1 = 0$.

Now we shift our attention to Bianchi cosmologies, i.e. Bianchi spacetimes
as defined in Definition~\ref{defn:BianchiSpacetime} below,
that moreover solve Einstein's equations for  fluid matter or in vacuum, or equivalently, 
developments arising from suitable left-invariant initial data on a Lie group.
Bianchi spacetimes have a three-dimensional group of isometries
acting freely on spacelike hypersurfaces, the surfaces of homogeneity.
In particular, there is a basis of spacetime Killing vector fields,
tangent to the group orbits and whose Lie algebra has the specified Bianchi type.
On the other hand, one may consider an orthonormal frame $\{e_0, e_1, e_2, e_3\}$ that commutes
with the Killing fields and so that $e_0$ is orthogonal to the group orbits.
The frame fields give rise to structure coefficients $\ga$ (not constants, as they vary over time),
which may be decomposed into $n$ and $a$,
similar to the decomposition of $\hg$ into $\hn$ and $\ha$ above.
These structure coefficients are then structure constants at every instance of time,
and the Lie algebra also has the specified Bianchi type.
We refer to Section~1.5 of \cite{DynSysCosmology} for details.

Now for Bianchi cosmologies that are not of type VIII or IX,
there is also an Abelian subgroup $G_2$ of the isometry group.
One may then specialize the orthonormal frame such that $e_2$ and $e_3$ lie tangent
to this Abelian subgroup $G_2$, cf. the appendix in \cite{WainwrightExact1981}.
In particular Bianchi cosmologies of all types except VIII and IX
are special cases of so-called \emph{$G_2$ cosmologies},
which are solutions to Einstein's field equations admitting an Abelian group of isometries 
whose orbits are spacelike two-surfaces, cf. Sections 1.6 and 12 of \cite{DynSysCosmology}.
$G_2$ cosmologies have been studied considerably
see in particular the work by Wainwright and collaborators
e.g. \cite{WainwrightClassification1979}, \cite{WainwrightExact1981},
\cite{HewittWainwrightOTG2} and \cite{UgglaWainwrightVanElstG2}
wherein geometrical classifications are given and equations
in the dynamical systems approach are formulated,
and more recently by Hewitt and collaborators,
\cite{HewittConjectureChaoticCosmInhom} and \cite{HewittEtAlHSOKVF},
wherein self-similar $G_2$ cosmologies are studied
and a conjecture about the late-time behaviour for $G_2$ cosmologies is stated.
Besides the Bianchi cosmologies of lower types,
also cosmologies which have Gowdy symmetry fall under the class of $G_2$ cosmologies,
cf. \cite{UgglaWainwrightVanElstG2, RingstromGowdyLivingReviews} for more details.

The $G_2$ may have two important properties,
following the classification of $G_2$ cosmologies by Wainwright \cite{WainwrightExact1981}.
It may act \emph{orthogonally-transitively},
which means that the planes orthogonal to the group orbits
are the tangent planes of some two-surface, i.e. are two-surface forming.
Moreover, one or two of the Killing vectors may be \emph{hypersurface-orthogonal}
which means that the tangent planes in spacetime that are orthogonal to 
this Killing vector are three-surface forming.

\subsection*{Orthogonal transitivity}
If we impose Einstein's field equations with an perfect fluid
on a Lie group of a lower Bianchi type with a left-invariant first and second fundamental form,
and the fluid congruence is orthogonal to the surfaces of homogeneity, 
then the momentum constraints imply that the $G_2$ acts orthogonally-transitively,
except if the Bianchi type is~\vi{}.
The $G_2$ acting orthogonally-transitively is, in fact, equivalent 
to the spacelike direction normal to the $G_2$, i.e. the direction parallel to $a^\#$,
being a shear eigenspace by Theorem~3.1(a) of \cite{WainwrightClassification1979}.
In particular, the tangent planes to the $G_2$ may be written as sums of 
eigenspaces of the shear or equivalently of the Weingarten map.
The $G_2$ being Abelian then implies 
the vanishing of a structure coefficient with respect to a frame
    diagonalizing the Weingarten map.
As explained in the introduction, this is particularly relevant 
for quiescence or more precisely for the convergence of the eigenvalues of the expansion-normalized Weingarten map.

For Bianchi type~\vi{} the momentum constraints do not imply that the $G_2$ acts 
    orthogonally-transitivity.
This is due to a certain degeneracy of the momentum constraints compared to the other Bianchi types,
this was already observed by Ellis and MacCallum in \cite{EllisMacCallumHomCos}.
As the $G_2$ need not act orthogonally-transitively,
Bianchi type~\vi{} orthogonal perfect fluid cosmologies have one more degree of freedom
compared to orthogonal perfect fluid cosmologies of the other lower Bianchi types.
This is because the vector $a^\#$ orthogonal to the $G_2$ is no longer an eigenvector of the Weingarten map,
which means that the planes tangent to the $G_2$ are no longer a sum of eigenspaces of the Weingarten map.

\subsection*{A spacelike hypersurface-orthogonal Killing vector field}
One of the Killing vector fields generating the isometry group 
    may be hypersurface-orthogonal.
This means that the planes tangent to it, which contain two spacelike directions 
and a timelike direction, are integrable in the sense of Frobenius.
In the literature this condition has sometimes been referred to as polarization
(though it should be said that this latter word has been used
    for various geometric features that reduce the 
gravitational degrees of freedom).
If one of the two spacelike Killing fields of a $G_2$ cosmology
is hypersurface-orthogonal,
then there is another hypersurface-orthogonal Killing vector
field orthogonal to the first one if and only if the $G_2$ acts orthogonally-transitively,
see Theorem 2.1 of \cite{WainwrightExact1981}.

A necessary condition for the existence of the hypersurface-orthogonal Killing vector field
in a type~\vi{} orthogonal perfect fluid cosmology
is that, with respect to a frame diagonalizing the shear, 
we must have $\ga_{31}^2 = n^{22} =0$, where $\xi = \xi_2 e_2$.
This frame also diagonalizes the expansion-normalized Weingarten map of course,
which is of interest for quiescence.
Furthermore, the hypersurface-orthogonal Killing vector field
    must lie tangent to the $G_2$, $\tr(n)$ must vanish 
and lastly one of possibly two hypersurface-orthogonal Killing vector fields must be 
a shear eigenvector with eigenvalue 0.
%
%
\section{Derivation of the evolution equations} \label{sec:AppendixEquations}
\noindent In this appendix we derive the evolution equations as presented 
in Section~\ref{sec:Equations} above.
The equations describe Bianchi type~\vi{} spacetimes
satisfying Einstein's equations for an orthogonal perfect fluid,
where we impose a linear equation of state of the form $p = (\gamma - 1) \rho$
with $\gamma \in (\tfrac{2}{3},2]$. 
The derivation of the equation goes in two steps.
First we choose a frame that in effect diagonalizes the structure coefficients,
and we expansion-normalize the variables and rescale time so as to obtain a 
dimensionless or scale-invariant system.
Then we introduce polar coordinates for the structure coefficients
and for the off-diagonal shear components, which allows us to write the Hamiltonian
and momentum (or Gauss and Codazzi) constraints as graphs,
thus eliminating some of the variables and constraints.
We end up with a system representing the five degrees of freedom in six variables and one constraint.

\subsection{Bianchi type VI$_{-1/9}$ cosmologies}
The spacetimes we consider in this work are assumed to satisfy Einstein's field equations
without cosmological constant with a perfect fluid matter source.
Thus, if $(M,g)$ denotes our Lorentz manifold, we have
\begin{equation}
\roRic_g - \tfrac{1}{2} \roScal_g g = T,
\end{equation}
where $\roRic_g$ and $\roScal_g$ are respectively the Ricci and scalar curvature associated
to the Levi-Civita connection of $g$,
and the energy-momentum tensor $T$ is that of an orthogonal perfect fluid
with a linear equation of state, which we specify below.
We moreover demand $(M,g)$ to be a Bianchi type~\vi{} spacetime.

\begin{defn} \label{defn:BianchiSpacetime}
A \emph{Bianchi type~\vi{} spacetime} is a Lorentz manifold $(M,g)$ of the form $M = I \times G$,
where $I$ is an open interval and $G$ a connected, three-dimensional Lie group,
with Lie algebra of type~\vi{}
and the metric may be written as 
\begin{equation}
g = - dt \otimes dt + a_{ij}(t) \xi^i \otimes \xi^j;
\end{equation}
here the $\xi^j$ are dual to a basis $e_i$ of the Lie algebra of $G$,
and the functions $a_{ij}$ are smooth and form the components 
of a symmetric, positive definite matrix at any $t \in I$.
\end{defn}

The energy momentum tensor for a perfect fluid may be written as 
\begin{equation}
T = (\rho + p) u \otimes u + p g
\end{equation}
where $\rho \geq 0$ denotes the energy density, $p$ the pressure,
and $u$ the four-velocity of the fluid.
We assume a linear equation of state, which means that for some constant $\gamma \geq 0$
we have $p = (\gamma - 1) \rho$. 
Typically one considers $\gamma \in (2/3,2]$;
for this range both the strong and dominant energy conditions are satisfied
A stiff fluid is a perfect fluid for which the pressure equals the energy density,
i.e. $\gamma = 2$,
while orthogonality means that the four-velocity of the fluid is orthogonal to the 
orbits of the Lie group $G$. 
In particular for orthogonal perfect fluid in a Bianchi spacetime the energy-momentum tensor takes the form
\begin{equation}
T = 2 \rho~ dt \otimes dt  + \rho g     
\end{equation}
where $\rho \geq 0$ is the energy density of the fluid.

\subsection{Expansion-normalized canonical coordinates}
Our starting point to derive the evolution equations are the Equations (1.90)-(1.99) 
    of \cite{DynSysCosmology}, where we write $\theta = 3 H$ for the mean curvature.
We choose to rotate the frame in such a way that the symmetric matrix $n$ is diagonalized,
and moreover so that $n^{11} = 0$, which we may do initially just as in the classification.
In the notation of \cite{DynSysCosmology}, this condition may be propagated by letting 
\begin{equation}
\Omega_1 = \s_{23} Q, \qquad 
\Omega_2 = \s_{31}, \qquad
\Omega_3 = -\s_{12}
\end{equation}
where $Q:= (n_{22} + n_{33}) / (n_{22} - n_{33})$.
We refer to this choice of frame as the \emph{canonical frame},
and the resulting expansion-normalized coordinates as \emph{canonical coordinates}.
Observe that $Q$ is well-defined due to the fact that $n_{22}$ and $n_{33}$,
being the eigenvalues of $n$, must be non-zero
and have opposite signs for a development of type~\vi{}.
    
We proceed to expansion-normalize the variables, i.e. divide by the mean-curvature\footnote{
Our choice of expansion-normalization is the convention 
    of \cite{RingstromBianchiIX},
as opposed to the convention where one normalizes with the Hubble parameter $H = \theta/3$.},
\begin{equation}
\S_{ij} := \frac{\s_{ij}}{\theta}, \quad
N_{ij} := \frac{n^{ij}}{\theta}, \quad
A_{k} := \frac{a_k}{\theta}, \quad
\Omega := \frac{3 \rho}{\theta^2}.
\end{equation}
and rescale the time-parameter $\tau = \tau(t)$
as $d\tau/dt = \theta/3$, where $\theta = \tr(k)$.
The fact that $\theta >0$ throughout the development follows from the Hamiltonian constraint,
before it is expansion-normalized,
given that the spatial scalar curvature is strictly negative, as it is for type~\vi{} cosmologies
on the surfaces of homogeneity.
To continue, we use the fact that the shear is trace free to define as usual
\begin{equation} \label{eq:SigmaPlusSigmaMinus}
\S_+ := - \tfrac{3}{2} \S_{11},
\quad \S_- := \tfrac{\sqrt{3}}{2} \left( \S_{22} - \S_{33} \right),
\end{equation}
and in a similar vein we define
\begin{equation}
N_{\pm} :=\tfrac{\sqrt{3}}{2} \left( N_{22} \pm N_{33} \right),
\quad A := 3 \sqrt{3} A_1.
\end{equation}
We have the following set of expansion-normalized evolution equations 
(where $'$ denotes a derivative with respect to the time-coordinate $\tau$)
\begin{subequations} \label{eq:CanonEvE}
\begin{align}
A' &= (q  + 2 \S_{+}) A, \\
N_{\pm}' &= (q + 2 \S_{+}) N_{\pm} 
    + 2 \sqrt{3} \S_{-} N_{\mp} \\
\S_{+}' &= -(2-q)  \S_{+} - 2 N_{-}^2
	+ 9 (\S_{12}^2 + \S_{31}^2), \\
\S_{-}' &= -(2-q)  \S_{-} - 2\sqrt{3} N_{+} N_{-} 
  - 3 \sqrt{3} \left(\S_{31}^2 - \S_{12}^2 \right)
  + 6\sqrt{3} \S_{23}^2 Q  , \\
\S_{23}' &= -(2-q)  \S_{23} - \tfrac{2}{3} A N_-   
  - 2 \sqrt{3} \S_{-} \S_{23} Q 
  + 6 \S_{12} \S_{31}, \\
\S_{31}' &= -(2-q) \S_{31} - \big( 3 \S_{+} - \sqrt{3} \S_{-} \big) \S_{31} 
	- 3 \S_{12} \S_{23} (1 + Q), \label{eq:EvES31} \\
\S_{12}' &= -(2-q) \S_{12} - \big( 3 \S_{+} + \sqrt{3} \S_{-}\big) \S_{12}
	- 3 \S_{31} \S_{23} (1 - Q), \label{eq:EvES12}
\intertext{and there is also the auxiliary equation}
\Omega' &= 2(q-q^*) \Omega.
\end{align}
\end{subequations}
Here $q, q^*$ and $Q$ are given by 
\begin{subequations}
\begin{align}
q &:= q^* \Omega + 2 \left(\S_+^2 + \S_-^2 +
    3 \left( \S_{12}^2 + \S_{23}^2 + \S_{31}^2 \right) \right), \\
q^*(\gamma) &:= \tfrac{1}{2} (3 \gamma - 2), \\
Q &:= N_+ / N_-,
\end{align}
\end{subequations}
where $\gamma \in (\tfrac{2}{3},2]$.
For reference, the Raychaudhuri equation may be written as 
\begin{equation}
\theta' = -(1+q)\theta,
\end{equation}
and $Q$ satisfies the evolution equation
\begin{equation} \label{eq:EvEQ}
Q' = 2 \sqrt{3} \S_- \big(1-Q^2 \big).
\end{equation}
Lastly there are the various constraint equations,
namely the algebraic constraint, the momentum (or Codazzi) constraints and
the Hamiltonian (or Gauss) constraint, 
which for our choice of gauge respectively read
\begin{subequations}
\begin{align}
A^2 &=  N_{-}^2 - N_{+}^2 , \label{eq:AC} \\
\tfrac{1}{3} \S_+ A &= \S_{23} N_-, \label{eq:MC1}\\
\S_{12} A &= \S_{31} \left(N_{-} - N_{+}\right), 
    \label{eq:MC2} \\
\S_{31} A &= \S_{12} \left(N_{-} + N_{+}\right), 
    \label{eq:MC3}\\
1 &= \Omega + \tfrac{1}{3} A^2 + N_-^2 + \S_{+}^2 + \S_-^2 
    + 3 \left(\S_{12}^ 2+ \S_{23}^2 + \S_{31}^2 \right). \label{eq:HCOld}
\end{align}
\end{subequations}
\begin{rmk} \label{rmk:VariableQ}
We may replace $N_+$ with $Q$ as an independent variable of the system above,
observing that we can write $N_+ = N_- Q$
and obtain an equivalent system. 
This viewpoint is of use in Subsection~\ref{sec:PolarCoordinates} below.
\end{rmk}

\subsection{Invariant subsets}
There are multiple subsets invariant under the flow of the system of equations above, 
which are related to possible geometrical features of the Abelian $G_2$
described in Subsection~\ref{sec:abelianG2}.
Firstly, there are the orbits corresponding to a cosmology in which
the $G_2$ acts orthogonally-transitively,
i.e. when $\S_{12} = \S_{31} = 0$. 
We can observe from Equations \eqref{eq:EvES31} and \eqref{eq:EvES12},
that this indeed forms an invariant set, which we denote by \ot.
We refer to these orbits as the \emph{orthogonally-transitively} orbits
or sometimes also the \emph{non-exceptional} orbits.
On the other hand, orbits which are not orthogonally-transitive or not non-exceptional
we refer to as \emph{exceptional} orbits
for historical reasons.

Next, there are the orbits corresponding to a cosmology
in which the $G_2$ has a hypersurface-orthogonal Killing vector field.
This is equivalent to the condition that 
\begin{equation}
\S_{31} = \pm \S_{12}, \quad 
\S_- = 0, \quad \text{and}~
N_+ = 0.
\end{equation}
We denote this invariant set by $\svi$,
and refer to orbits of this type as \emph{polarized orbits}.
Analogous sets have been called \emph{shear-invariant}
e.g. in \cite{WainwrightHsuClassA}.

\subsection{Equivalent strata} \label{sec:strata}
Let us first observe that there are various equivalent strata of the space defined by the various variables and constraints.
In particular, if we map 
\begin{equation*}
(A, N_+,N_-, \S_{12}, \S_{31}) \mapsto -(A, N_+, N_-, \S_{12}, \S_{31}),    
\end{equation*}
then the dynamics
and constraint equations are invariant,
and similarly if we map 
\begin{equation*}
(N_+, N_-, \S_{12}, \S_{23}) \mapsto -(N_+, N_-, \S_{12}, \S_{23}),    
\end{equation*}
or if we map 
\begin{equation*}
(N_+, N_-, \S_{23}, \S_{31}) \mapsto - (N_+, N_-, \S_{23}, \S_{31}).
\end{equation*}
This is a result of the discrete symmetries that we can apply to the frame,
namely the maps $S_j: e_j \mapsto -e_j$, where $j = 1,2,3$ respectively.
Lastly, we can choose to rotate in the $G_2$, say $R_{23}: (e_2, e_3) \mapsto (e_3, -e_2)$.
Then we obtain the map 
\begin{equation*}
(N_-,  \S_-, \S_{23}, \S_{31}, \S_{12},) \mapsto -(N_-, \S_-, \S_{23}, \S_{12}, -\S_{31}),
\end{equation*}
using the fact that the matrix $\S_{ij}$ is symmetric.
As for example $S_3 = R_{23} \circ S_2 \circ R_{23}$,
we see that we have $16$ equivalent strata of the phase space.

Hence, we choose to let 
\begin{equation} \label{eq:PreferredStratum}
A > 0, \quad 
N_- > 0, \quad 
\S_{12} \geq 0, \quad 
\S_{31} \geq 0.    
\end{equation}
Indeed, we may first use the reflection map $S_1$ to always obtain $A > 0$,
and one of $S_2, S_3$ or $R_{23}$ to ensure that $N_- > 0$.
Then, by \eqref{eq:MC2}, $\S_{12}$ and $\S_{13}$ must have the same sign,
so if they both have a minus sign then we may use $R_{23}^2 = S_2 \circ S_3$ to obtain 
that $\S_{12} \geq 0$ and $\S_{13} \geq 0$.

Observe that the non-exceptional orbits satisfy $\S_{12} = \S_{31} = 0$,
in which case $S_2$ and $S_3$ describe the same symmetry,
or in other words, $R_{23}^2 = \Id$ for the non-exceptional orbits.
In that case we only have 8 different strata.

\subsection{Adapted polar coordinates} \label{sec:PolarCoordinates}
In canonical coordinates with
    $N_{11} = 0$, the variables $N_{22}$ and $N_{33}$ are
the non-zero eigenvalues of the matrix $N$.
Thus $N_+$ and $N_-$ are up to a factor the sum and difference of the eigenvalues.
By the algebraic constraint, $A^2$ is then a multiple of the product of the eigenvalues.
Moreover, upon observation of the momentum constraints
it becomes clear that $N_{22}/N_{33}$,
i.e. the ratio of the eigenvalues, plays a fundamental role in the dynamics.

This ratio can be captured effectively by introducing polar coordinates
for the pairs $N_+, A$ and $\S_{12}, \S_{31}$.
Define the coordinates $ M \in (0,\infty)$, $P \in [0,\infty)$,
$\theta \in (0, \pi)$, and $\phi \in [0,\pi/2]$ as follows:
\begin{subequations}
\begin{align}
N_+ = M \cos(\theta), \\
A = M \sin (\theta), \\
\S_{31} = \tfrac{1}{\sqrt{3}} P \cos (\phi), \\
\S_{12} = \tfrac{1}{\sqrt{3}} P \sin (\phi).
\end{align}
\end{subequations}
Observe that the restrictions on $\theta$ and $\phi$ are a consequence
of choosing the stratum such that $A > 0, \S_{12} \geq 0, \S_{31} \geq 0$.
In particular $\sin(\theta) > 0$.
We also let $\S_{\times} := \sqrt{3} \S_{23}$.
We refer to this set of coordinates as \emph{adapted polar coordinates}.

In these coordinates, we firstly find that
$N_- = M$, again using the assumption $N_- > 0$,
and secondly that Equations \eqref{eq:MC2} and \eqref{eq:MC3} may be simply written as  
\begin{subequations}
\begin{align}
M P ( \cos(\theta - \phi) -\cos(\phi) ) &= 0, \\
M P ( \sin(\theta - \phi) -\sin(\phi) ) &= 0,
\end{align}
\end{subequations}
which means that, either $M = 0$ or $P = 0$ or $\phi = \theta/2 $.
The first condition means the corresponding orbit is of Bianchi type I, 
while the second condition means that the orbit 
in the invariant set of non-exceptional orbits $\otp$ for which $P$ vanishes identically.
The condition $\phi = \theta /2$ is of most interest, as this is the only possible solution
for the generic set of exceptional orbits, and allows us to eliminate the angle $\phi$
inside the phase space;
for the non-exceptional orbits the angle $\phi$ does not play a role in the dynamics
and we may simply set it to $\theta/2$ regardless.

Observe that in these coordinates $Q = \cos(\theta)$,
and we have now effectively taken the viewpoint of Remark~\ref{rmk:VariableQ}
and turned the ratio $N_+/N_-$ into a variable.

The momentum constraint \eqref{eq:MC1} in adapted polar coordinates reads
\begin{equation}
M \big( \S_+ \sin(\theta) - \sqrt{3} \S_{\times}) = 0
\end{equation}
and allows us to write $\S_{\times} = \frac{1}{\sqrt{3}} \S_+ \sin(\theta)$,
and hence to eliminate $\S_{\times}$ as a variable.
In what follows we also eliminate $\Omega$ using the Hamiltonian constraint \eqref{eq:HCOld},
rewritten in adapted polar coordinates, to write $\Omega$ as a function of the other variables.

We thus have the following set of evolution equations:
\begin{subequations} \label{eq:PolarEvEAppendix}
\begin{align} 
M' &= \big[q +2 \S_{+} + 2 \sqrt{3} \cos(\theta) \S_-  \big] M, \label{eq:EvEMApp}\\
P' &= \big[-(2-q) - \big(3 + \sin^2(\theta) \big) \S_{+}  
    + \sqrt{3} \cos(\theta) \S_-  \big] P, \label{eq:EvEPApp} \\
\theta' &= -2 \sqrt{3} \S_- \sin(\theta), \label{eq:EvEThApp}\\
\S_{+}' &= -(2-q) \S_{+} - 2 M^2 + 3 P^2, \label{eq:EvESplusApp} \\
\S_{-}' &= -(2-q) \S_- - \sqrt{3} \cos(\theta) \big( 2 M^2 + P^2 
    - \tfrac{2}{3} \sin^2(\theta) \S_{+}^2 \big) \label{eq:EvESminusApp},
\intertext{where $q$ is shorthand for}
q &= (2 - q^*) \big( \big(1+ \tfrac{1}{3}\sin^2(\theta)\big) \S_+^2
    + \S_-^2 + P^2 \big) + q^* \left(1 - \big(1+ \tfrac{1}{3}\sin^2(\theta)\big) M^2 \right),
    \label{eq:DefinitionOfqApp}
\end{align}
\end{subequations}
where we eliminated $\Omega$ as a variable using the Hamiltonian constraint.

The constraints, which are now the definitions of $\Omega$ and $\S_{\times}$,
in these coordinates take the form 
\begin{subequations} \label{eq:AuxiliaryAppendix}
\begin{align} 
\Omega &= 1 - \big(1+ \tfrac{1}{3}\sin^2(\theta)\big) \big( M^2 + \S_+^2 \big) 
    - P^2 - \S_-^2, \label{eq:HCNewApp} \\
\S_{\times} &= \tfrac{1}{\sqrt{3}} \S_+ \sin (\theta), \label{eq:MCSxApp}
\intertext{and $\Omega$ and $\S_{\times}$ satisfy the so-called \emph{auxiliary evolution equations},
which are now simply a consequence of differentiating the constraints}
\S_{\times}' &= [-(2-q) - 2 \sqrt{3} \S_- \cos(\theta)] \S_{\times} 
    - \tfrac{1}{\sqrt{3}} \sin(\theta) \big( 2 M^2 - 3 P^2 \big), \\
\Omega' &= 2(q - q^*) \Omega. \label{eq:EvEOmApp}   
\end{align}
\end{subequations}

\begin{rmk}
Alternatively, having derived the system of equations (\ref{eq:EvEM} -~\ref{eq:EvESminus}),
one could substitute back $Q = \cos(\theta)$ as a variable, 
as only $\cos(\theta)$ and $\sin^2 (\theta)$
appears in the system of equations, except for the equation for $\theta$
that would be substituted by Equation \eqref{eq:EvEQ},
and for the constraint defining $\S_{\times}$. 
This would yield again an equivalent system, cf. Remark~\ref{rmk:VariableQ}.
This has the drawback of introducing square roots in the constraint defining $\S_\times$
as well as its evolution equation however, which is why we prefer to proceed with 
the angular coordinate $\theta$.
\end{rmk}

\subsection{Vacuum fixed points} \label{sec:VacuumFixedPoints}
There are also fixed points of the flow for the vacuum orbits, i.e. $\Omega = 0$,
which are described for example in \cite{HewittHorwoodWainrightExceptional}.
These do not play a large role in our analysis, but for the convenience of
the reader we decided to present them in adapted polar coordinates in our choice of stratum.
The vacuum fixed points play a role in the analysis of the future asymptotics,
which are briefly discussed in Section~\ref{sec:future}.
Of course there is the Kasner circle $K = \partial \scr{D}$,
corresponding to the Kasner solutions, which are well known and we not elaborate on here.
Next, there is the set of plane-wave equilibria $\mathrm{PW}$
(there are more sets but because of our choice of stratum we only have one in the phase space)
which in adapted polar coordinates may be written as 
\begin{equation}
\begin{split}
\mathrm{PW} = \big\{ (\sqrt{-\S_+(1+\S_+)}, 0 , \arcsin(\sqrt{\tfrac{3(1+\S_+)}{-\S_+}}), \S_+, 0)~|~
    \S_+ \in (-1,-\tfrac{3}{4}] \big\}.
\end{split}
\end{equation}
Note that $\mathrm{PW} \subset \otp^{0}$, and point with the value $\S_+ = -1$ corresponds to the special point $T_1$,
while the point with the value $\S_+ = -3/4$ lies in $\spvi^0$, and in \cite{HewittTraceClass}
is called the Lifschitz-Khalatnikov point.
In \cite{HewittWainwrightClassB} the analysis for the orthogonally transitive orbits
makes use of a monotone function
which in adapted polar coordinates reads 
\begin{equation}
(1+\S_+)^2 - \tfrac{4}{3} M^2 \sin^2(\theta).
\end{equation}
Using this function it is shown that $\mathrm{PW} \cup \{T_1\}$ 
forms the global future attractor
for orbits in $\otp^0$ as well as for orbits in $\otpo$ for $\gamma = 2$, aside of a set of measure zero.
Here $T_1$ is the special or Taub point on the Kasner circle
for which $s_+ = 1$ and $s_- = 0$ which corresponds to a flat Kasner solution.
However, as described in \cite{HewittHorwoodWainrightExceptional},
the extra degree of freedom in the shear
which is present in the exceptional orbit destabilizes the set of plane-wave solutions.

Lastly there is the Robinson-Trautman solution $\mathrm{RT}$
which in adapted polar coordinates reads
\begin{equation}
\mathrm{RT} = (\tfrac{1}{\sqrt{2}}, \tfrac{\sqrt{5}}{3 \sqrt{3}}, \pi/2, -\tfrac{1}{3}, 0).
\end{equation}
The point $\mathrm{RT}$ lies in the invariant set $\spvi^0$.
As $P(\mathrm{RT}) > 0$, it lies in the set of exceptional orbits.
For $\gamma \in (9/10,2)$ it is conjectured in \cite{HewittHorwoodWainrightExceptional}
that $\mathrm{RT}$ is the global future attractor.
The same reasoning can be extended to the case $\gamma = 2$,
and we indeed expect $\mathrm{RT}$ to be the global future attractor
for the stiff fluid case as well.
\section{Analytical tools} \label{sec:dynsys}
\noindent An important tool in the dynamical systems approach to cosmology
    is the monotonicity principle.
The proof of the lemma below can be found e.g. in appendix A of \cite{LeBlancKerrWainwrightMagneticVI0}.
\begin{lemma}[Monotonicity principle]
\label{lemma:monotonicity}
Let $\vp^\tau$ be a flow on $\bb{R}^n$ and $S$ an invariant set of $\vp^\tau$.
Let $Z: S \to \bb{R}$ be a $C^1$-function whose range is the interval $(a,b)$,
where $a \in \bb{R} \cup \{-\infty\}, b \in \bb{R} \cup \{\infty\}$ and $a<b$. 
If $Z$ is strictly monotonically decreasing on orbits in $S$, then for all $x \in S$,
\begin{align}
\begin{split}
\alpha(x) &\subseteq 
\left\{s \in \bar{S}\setminus S ~|~ \lim_{y \to s} Z(y) \neq a \right\}, \\
\omega(x) &\subseteq 
\left\{s \in \bar{S}\setminus S~ |~ \lim_{y \to s} Z(y) \neq b \right\}. 
\end{split}
\end{align}
\end{lemma}
\noindent Another tool of use is an elementary version of 
    Gr{\"o}nwall's lemma, see e.g. Lemma~7.1 of \cite{RingstromCauchy}
for a more complete statement and proof. 
Note the reversed orientation of time.
\begin{lemma}
\label{lemma:Gronwall}
Let $I:=(a,b)$ be in an interval with $-\infty \leq a < b < \infty$.
Let $u \in C^0(\bar{I})\cap C^1(I)$ and $\alpha \in C(\bar{I})$.
If $u$ satisfies the inequality $u'(\tau) \leq - \alpha(\tau) u(\tau)$,
then
\begin{equation}
u(\tau) \geq u(b) \exp \left(\int_\tau^b \alpha(s) ds \right)
\end{equation}
for all $\tau \in I$.
\end{lemma}

Next, we have a practical lemma, which we did not find a reference for
in the literature, so we include a proof for completeness.
Below $C_b^1(I)$ denotes the space of continuously differentiable
functions with bounded derivative.
\begin{lemma} \label{lemma:EdgeCases}
Let $I:=(-\infty,b]$, and let $u,v,w \in C_b^1(I)$,
$\alpha, \beta \in C^1_b(I) \cap L^1(I)$ be such that: \\ 
$u \geq 0, u(b) > 0, |w| \leq \frac{1}{2}, \beta \geq 0$,
and 
\begin{equation}
\limtau u(\tau) = 0, \quad \limtau v(\tau) = v_0, 
\end{equation}
and such that the following differential inequalities are satisfied:
\begin{align}
u'(\tau) &\leq (-\alpha(\tau) + (1+w(\tau)) v(\tau)) u(\tau), \\
v'(\tau) &\leq -\beta(\tau) v (\tau) + f(\tau) - c u(\tau)^m
\end{align}
where $c > 0$ is a constant, $m \in \mathbb{N}$ is an integer,
and, finally, $f$ is bounded from above by 
a function converging to 0 exponentially as $\tau \to -\infty$,
i.e. there exists a $T_f$ such that 
$f(\tau) \leq f_0 \exp(\delta \tau)$ for some constants $f_0 \geq 0$, $\delta > 0$,
for any $\tau \leq T_f$.
Then $v_0 > 0$.
\end{lemma}
\begin{proof}
Firstly, let us observe that for any $\tau \leq b$
\begin{equation} 
u (\tau) \geq u(b) \exp( \int_{\tau}^{b} \left(\alpha(s) - (1+w(s)) v(s)\right) ds) , 
\end{equation}
as a consequence of Lemma~\ref{lemma:Gronwall}.
Hence as $u(\tau) \to 0, u(b) >0$ and 
$\limtau \int_{\tau}^b \alpha(s) ds$ is finite,
we must have that $\int_{\tau}^{b} (1+w(s)) v(s) ds \to \infty$ as $\tau \to -\infty$.
However, as $|w| \leq \frac{1}{2}$, this means that also 
the integral $\int_{\tau}^{b} v(s) ds \to \infty$ as $\tau \to -\infty$.
It follows that $v_0 \geq 0$, and there must also exist a series of
times $(\tau_k)_{k=1}^\infty$ such that $\tau_k \to -\infty$ and such that $v(\tau_k) >0$.

Assume now to the contrary that $v_0 = 0$.
Let $T_f, f_0$ and $\delta$ be as given in the statement.
The same statements in the paragraph above hold with $b$ replaced by $T_f$,
by running the same argument ending at the time $T_f$,
where we note in particular that $u(T_f) >0$ by the inequality shown above.

On the one hand, there must exist a $T_\delta$ such that $|v(\tau)| \leq \delta/(3m)$
    for any $\tau \leq T_\delta$, as $v (\tau) \to 0$.
On the other hand, there must exist a time $T \leq \min\{T_{\delta},T_f\}$ such that 
for any $\tau \leq T$ it must hold that
\begin{equation} \label{eq:InequalityForF}
    f(\tau) < \tfrac{c}{2} u(\tau)^m.
\end{equation}
For if not, by comparing logarithms, we could deduce the existence
of a series of times $(\rho_k)_k$ such that  $\rho_k \to -\infty$ and such that for each 
$k \in \mathbb{N}$ it would hold that
\begin{equation*}
\delta \rho_k + m \int_{\rho_k}^{T_f} (1+w(s)) v(s) ds \geq - \ln(f_0) 
+ \ln(c/2) + m \ln(u(T_f)) + m \int_{\rho_k}^{T_f} \alpha(s) ds.
\end{equation*}
But as $|1+|w(s)||v(\tau)|\leq \delta/(2m)$ for $\tau \leq T_f$, that would lead to a contradiction,
as all terms on the right-hand side are bounded on $(-\infty,b]$,
while the terms on the left-hand side are diverging to $-\infty$ as $k \to \infty$.

Recalling the differential inequality for $v$,
\eqref{eq:InequalityForF} implies that $v'(\tau) < 0$ for any $\tau \leq T$
such that $v(\tau) >0$, as we assumed that $\beta \geq 0$.
Thus, if $v(\tau) > 0$ at some time $\tau \leq T$, then it is also increasing towards the past.
As we know that along the sequence of times $(\tau_k)_k, \tau_k \to -\infty$
it holds that $v(\tau_k) > 0$, we also know that there exists a $k_0 \in \mathbb{N}$
such that $\tau_{k_0} \leq T$,
and we may conclude that $v(\tau) \geq v(\tau_{k_0})$ for any $\tau \leq \tau_{k_0}$.
Hence we must have $v_0 = \limtau v(\tau) > 0$.
\end{proof}

\printbibliography

@article{WainwrightHsuClassA,
  title={A dynamical systems approach to Bianchi cosmologies: orthogonal models of class A.},
  author={Wainwright, J. and Hsu, L.},
  journal={Classical and Quantum Gravity},
  volume={6},
  number={10},
  pages={1409--1431},
  year={1989},
  publisher={IOP Publishing},
  doi={10.1088/0264-9381/6/10/011}
}

@article{RingstromBianchiIX,
  title={The Bianchi IX Attractor.},
  author={Ringstr{\"o}m, H.},
  journal={Annales Henri Poincar{\'e}},
  volume={2},
  pages={405--500},
  year={2001},
  publisher={Springer},
  doi={10.1007/PL00001041}
}

@misc{RingstromInitialDataSingSym,
  title={Initial data on big bang singularities in symmetric settings.},
  author={Ringstr{\"o}m, H.},
  year={2022},
  eprint={2202.11458},
  archivePrefix={arXiv},
  primaryClass={gr-qc}
}

@article{RadermacherVacuum,
  title={Strong Cosmic Censorship in Orthogonal Bianchi Class B Perfect Fluids and Vacuum Models.},
  author={Radermacher, K.},
  journal={Annales Henri Poincar{\'e}},
  volume={20},
  pages={689--796},
  year={2019},
  publisher={Springer},
  DOI={10.1007/s00023-018-00756-1}
}

@misc{RadermacherStiff,
  title={Orthogonal Bianchi B stiff fluids close to the initial singularity.},
  author={Radermacher, K.},
  archivePrefix={arXiv},
  year={2017},
  eprint={1712.02699},
  primaryClass={gr-qc}
}

@article{HewittWainwrightClassB,
  title={A dynamical systems approach to Bianchi cosmologies: orthogonal models of class B.},
  author={Hewitt, C.G. and Wainwright, J.},
  journal={Classical and Quantum Gravity},
  volume={10},
  number={1},
  pages={99--124},
  year={1993},
  doi={10.1088/0264-9381/10/1/012}
}

@book{DynSysCosmology,
  title={Dynamical Systems in Cosmology.},
  author={Wainwright, J. and Ellis, G.F.R.},
  publisher={Cambridge University Press},
  year={1997},
  doi={10.1017/CBO9780511524660},
}

@article{HewittTraceClass,
  title={An investigation of the dynamical evolution of a class of Bianchi VI$_{-1/9}$ cosmological models.},
  author={Hewitt, C.G.},
  journal={General relativity \& Gravitation},
  volume={23},
  pages={691--712},
  year={1991},
  publisher={Springer},
  doi={10.1007/BF00756774}
}

@article{HewittHorwoodWainrightExceptional,
  title={Asymptotic dynamics of the exceptional Bianchi cosmologies.},
  author={Hewitt, C.G. and Horwood, J.T. and Wainwright, J.},
  journal={Classical and Quantum Gravity},
  volume={20},
  number={9},
  pages={1743--1756},
  year={2003},
  doi={10.1088/0264-9381/20/9/311}
}

@article{EllisMacCallumHomCos,
  title={A Class of Homogeneous Cosmological Models.},
  author={Ellis, G.F.R. and MacCallum, M.A.H.},
  journal={Communications in Mathematical Physics},
  volume={12},
  pages={108--141},
  year={1969},
  publisher={Springer},
  doi={10.1007/BF01645908}
}

@article{CollinsHawkingIsotropic,
  title={Why is the universe isotropic?},
  author={Collins, C.B. and Hawking, S.W.},
  journal={Astrophysical Journal},
  volume={180},
  pages={317--334},
  year={1973},
  doi={10.1086/151965}
}

@article{LeBlancKerrWainwrightMagneticVI0,
  title={Asymptotic states of magnetic Bianchi VI$_0$ cosmologies.},
  author={LeBlanc, V.G. and Kerr, D. and Wainwright, J.},
  journal={Classical and Quantum Gravity},
  volume={12},
  number={2},
  pages={513},
  year={1995},
  doi={10.1088/0264-9381/12/2/020}
}

@article{IsenbergMoncriefHalfPolarized,
  doi={10.1088/0264-9381/19/21/305},
  year={2002},
  volume={19},
  number={21},
  pages={5361--5386},
  author={Isenberg, J. and Moncrief, V.},
  title={Asymptotic behaviour in polarized and half-polarized $U(1)$ symmetric vacuum spacetimes.},
  journal={Classical and Quantum Gravity}
}

@article{OudeGroenigerVI0,
  doi={10.1007/s00023-020-00934-0},
  year={2020},
  publisher={Springer Science},
  volume={21},
  number={9},
  pages={3069--3094},
  author={Oude Groeniger, H.},
  title={On Bianchi Type VI$_0$ Spacetimes with Orthogonal Perfect Fluid Matter.},
  journal={Annales Henri Poincar{\'{e}}}
}

@article{HeinzleUgglaNew,
  doi={10.1088/0264-9381/26/7/075015},
  year={2009},
  volume={26},
  number={7},
  pages={075015},
  author={Heinzle, J.M. and Uggla, C.},
  title={A new proof of the Bianchi type {IX} attractor theorem.},
  journal={Classical and Quantum Gravity}
}

@article{BeguinDutilleul,
  doi={10.1007/s00220-022-04583-8},
  year={2023},
  volume={399},
  number={2},
  pages={737--927},
  author={B{\'{e}}guin, F. and Dutilleul, T.},
  title={Chaotic Dynamics of Spatially Homogeneous Spacetimes.},
  journal={Communications in Mathematical Physics}
}

@article{RendallMixmaster,
  doi = {10.1088/0264-9381/14/8/028},
  year = {1997},
  volume = {14},
  number = {8},
  pages = {2341--2356},
  author = {Rendall, A.D.},
  title = {Global dynamics of the mixmaster model.},
  journal = {Classical and Quantum Gravity}
}

@article{WeaverMagnetic,
  doi = {10.1088/0264-9381/17/2/311},
  year = {1999},
  volume = {17},
  number = {2},
  pages = {421--434},
  author = {Weaver, M.},
  title = {Dynamics of magnetic Bianchi VI$_0$ cosmologies.},
  journal = {Classical and Quantum Gravity}
}

@article{FenichelGeomSingPert,
  title={Geometric singular perturbation theory for ordinary differential equations.},
  author={Fenichel, N.},
  journal={Journal of differential equations},
  volume={31},
  number={1},
  pages={53--98},
  year={1979},
  publisher={Academic Press},
  doi={10.1016/0022-0396(79)90152-9}
}

@book{WigginsNonlinearDynSys,
  title={Introduction to Applied Nonlinear Dynamical Systems and Chaos.},
  author={Wiggins, S.},
  year={1990},
  publisher={Springer Verlag},
  series={Texts in Applied Mathematics},
  volume={2},
  doi={10.1007/978-1-4757-4067-7}
}

@book{LeeSmoothManifolds,
  title={Introduction to Smooth Manifolds.},
  author={Lee, J. M.},
  year={2012},
  publisher={Springer},
  doi={10.1007/978-1-4419-9982-5}
}

@article{HewittConjectureChaoticCosmInhom,
  title={Dynamical equilibrium states of a class of irrotational non-orthogonally transitive $G_2$ cosmologies I: The conjecture of chaotic cosmological inhomogeneity},
  author={Hewitt, C.G.},
  journal={Classical and Quantum Gravity},
  volume={38},
  number={17},
  pages={175004},
  year={2021},
  doi={10.1088/1361-6382/ac0e46}
}

@article{HewittEtAlHSOKVF,
  title={Dynamical equilibrium states of a class of irrotational non-orthogonally transitive $G_2$ cosmologies: II. Models with one hypersurface-orthogonal Killing vector field.},
  author={Rashidi, S. and Hewitt, C.G. and Charbonneau, B.},
  journal={Classical and Quantum Gravity},
  volume={38},
  number={15},
  pages={155013},
  year={2021},
  doi={10.1088/1361-6382/ac0bc0}
}

@article{WainwrightExact1981,
  title={Exact spatially inhomogeneous cosmologies.},
  author={Wainwright, J.},
  journal={Journal of Physics A: Mathematical and General},
  volume={14},
  number={5},
  pages={1131--1147},
  year={1981},
  doi={10.1088/0305-4470/14/5/033}
}

@article{WainwrightClassification1979,
  title={A classification scheme for non-rotating inhomogeneous cosmologies.},
  author={Wainwright, J.},
  journal={Journal of Physics A: Mathematical and General},
  volume={12},
  number={11},
  pages={2015-2029},
  year={1979},
  doi={10.1088/0305-4470/12/11/014}
}

@article{HewittWainwrightOTG2,
  title={Orthogonally transitive $G_2$ cosmologies.},
  author={Hewitt, C.G. and Wainwright, J.},
  journal={Classical and Quantum Gravity},
  volume={7},
  number={12},
  pages={2295},
  year={1990},
  doi={10.1088/0264-9381/7/12/011}
}

@article{UgglaWainwrightVanElstG2,
  title={Dynamical systems approach to $G_2$ cosmology.},
  author={{Van Elst}, H. and Uggla, C. and Wainwright, J.},
  journal={Classical and Quantum Gravity},
  volume={19},
  number={1},
  pages={51},
  year={2001},
  doi={10.1088/0264-9381/19/1/304}
}

@article{RingstromGowdyLivingReviews,
  doi = {10.12942/lrr-2010-2},
  year = {2010},
  volume = {13},
  number = {1},
  author = {Ringstr\"{o}m, H.},
  title = {Cosmic Censorship for Gowdy Spacetimes.},
  journal = {Living Reviews in Relativity}
}

@book{RingstromCauchy,
  title={The Cauchy Problem in General Relativity.},
  author={Ringström, H.},
  year={2009},
  publisher={European Mathematical Society},
  series={ESI Lectures in Mathematics and Physics}
}

@article{AnderssonRendall,
  doi={10.1007/s002200100406},
  year={2001},
  publisher={Springer Science and Business Media {LLC}},
  volume={218},
  number={3},
  pages={479--511},
  author={Andersson, L. and Rendall, A.D.},
  title={Quiescent Cosmological Singularities.},
  journal={Communications in Mathematical Physics}
}

@article{BKLTimeSingularity1982,
  title={A general solution of the Einstein equations with a time singularity.},
  author={Belinski, V.A. and Khalatnikov, I.M. and Lifshitz, E.M.},
  journal={Advances in Physics},
  volume={31},
  number={6},
  pages={639--667},
  year={1982},
  doi={10.1080/00018738200101428}
}

@misc{RingstromGeometrySilentBigBang,
  author={Ringstr\"{o}m,  H.},
  title={On the geometry of silent and anisotropic big bang singularities.},
  year={2021},
  eprint={2101.04955},
  archivePrefix={arXiv},
  primaryClass={gr-qc}
}

@article{FournodavlosRodnianskiSpeck,
  doi={10.1090/jams/1015},
  year={2023},
  publisher={American Mathematical Society ({AMS})},
  volume={36},
  number={3},
  pages={827--916},
  author={Fournodavlos, G. and Rodnianski, I. and Speck, J.},
  title={Stable Big Bang formation for Einstein's equations: The complete sub-critical regime.},
  journal={Journal of the American Mathematical Society},
}

@misc{QuiescentBigBang,
  archivePrefix={arXiv},
  eprint={2309.11370},
  author={Oude Groeniger, H. and Petersen,  O. and Ringstr\"{o}m,  H.},
  title={Formation of quiescent big bang singularities.},
  year={2023},
  primaryClass={gr-qc}
}

@misc{RingstromInitialDataSingularity,
  archivePrefix={arXiv},
  eprint={2202.04919},
  author={Ringstr\"{o}m,  H.},
  title={Initial data on big bang singularities.},
  year={2022},
  primaryClass={gr-qc}
}

@article{FournodavlosLukAsymptoticallyKasner,
  doi={10.1353/ajm.2023.a902957},
  year={2023},
  volume={145},
  number={4},
  pages={1183--1272},
  author={Fournodavlos, G. and Luk, J.},
  title={Asymptotically Kasner-like singularities.},
  journal={American Journal of Mathematics},
}

@misc{FajmanUrbanBianchiV,
  archivePrefix={arXiv},
  eprint={2211.08052},
  author={Fajman, D. and Urban, L.},
  title={Cosmic Censorship near FLRW spacetimes with negative spatial curvature.},
  year={2022},
  primaryClass={gr-qc}
}

@misc{BeyerOliynykPastStabilityFluids,
  archivePrefix={arXiv},
  eprint={2308.07475},
  author={Beyer, F. and Oliynyk, T.A.},
  title={Past stability of FLRW solutions to the Einstein-Euler-scalar field equations and their big bang singularites.},
  year={2023},
  primaryClass={gr-qc}
}

@misc{BeyerOliynykPastStabilityScalarField,
  archivePrefix={arXiv},
  eprint={2112.07730},
  author={Beyer, F. and Oliynyk, T.A.},
  title={Localized big bang stability for the Einstein-scalar field equations.},
  year={2021},
  primaryClass={gr-qc}
}

@article{RodnianskiSpeckNearFLRW,
  doi={10.1007/s00029-018-0437-8},
  year={2018},
  volume={24},
  number={5},
  pages={4293--4459},
  author={Rodnianski, I. and Speck, J.},
  title={Stable Big Bang formation in near-{FLRW} solutions to the Einstein-scalar field and Einstein-stiff fluid systems.},
  journal={Selecta Mathematica},
}

@article{RodnianskiSpeckModeratelyAnisotropic,
  doi={10.4171/jems/1092},
  year={2021},
  volume={24},
  number={1},
  pages={167--263},
  author={Rodnianski, I. and Speck, J.},
  title={On the nature of Hawking's incompleteness for the Einstein-vacuum equations: The regime of moderately spatially anisotropic initial data.},
  journal={Journal of the European Mathematical Society},
}

@article{SpeckIsotropicBianchiIX,
  doi={10.1007/s00220-018-3272-z},
  year={2018},
  volume={364},
  number={3},
  pages={879--979},
  author={Speck, J.},
  title={The Maximal Development of Near-{FLRW} Data for the Einstein-Scalar Field System with Spatial Topology $\mathbb{S}^3$.},
  journal={Communications in Mathematical Physics},
}

@article{HewittBridsonWainwrightTiltedII,
  title={The Asymptotic Regimes of Tilted Bianchi II Cosmologies.},
  author={Hewitt, C.G. and Bridson, R. and Wainwright, J.},
  journal={General Relativity and Gravitation},
  volume={33},
  number={1},
  pages={65--94},
  year={2001},
  publisher={Springer},
  doi={10.1023/A:1002075902953}
}

@article{HervikTiltedVI0,
  title={The asymptotic behaviour of tilted Bianchi type VI$_0$ universes.},
  author={Hervik, S.},
  journal={Classical and Quantum Gravity},
  volume={21},
  number={9},
  pages={2301-2317},
  year={2004},
  doi={10.1088/0264-9381/21/9/007},
}
\end{document}